\documentclass[a4paper, 11pt]{article}
\usepackage[margin=3cm]{geometry}
\usepackage[utf8]{inputenc}
\usepackage[T1]{fontenc}
\usepackage{yfonts}
\usepackage{dsfont}
\usepackage{amsfonts,amsmath,amsxtra,amssymb,latexsym,amscd,amsthm,pb-diagram,fancyhdr,euscript}
\usepackage{latexsym}
\usepackage{mathrsfs}
\usepackage{mathtools}
\usepackage{multicol}
\usepackage{setspace}
\usepackage{slashed}
\usepackage[linktocpage, colorlinks=true, linkcolor={blue}, citecolor={blue}]{hyperref}
\usepackage{cleveref}
\usepackage{fancyhdr}
\usepackage{physics}
\usepackage{graphicx}
\usepackage{enumitem}
\usepackage[font=small,labelfont=bf]{caption}
\graphicspath{{./Figures/}}
\allowdisplaybreaks
\usepackage[sortcites,backend=bibtex,style=numeric,sorting=anyvt,doi=false,url=false,giveninits=true,isbn=false,maxbibnames=99]{biblatex}
\AtEveryBibitem{\clearfield{note}}
\AtEveryBibitem{\clearfield{month}}
\AtEveryBibitem{\clearfield{day}}

\title{\textbf{Scattering of Dirac Fields in the Interior of Kerr-Newman(-Anti)-de Sitter Black Holes via Comparison and Symmetry Operators}}
\author{\textbf{Mokdad Mokdad\thanks{University Grenoble Alps - Fourier Institute;  remi-mokdad.mokdad@univ-grenoble-alpes.fr} { }~~ and  Milos Provci\thanks{University of Münster; milos.provci@uni-muenster.de}}}
\date{}

\theoremstyle{theorem}
\newtheorem{theorem}{Theorem}[section]

\theoremstyle{theorem}
\newtheorem{definition}{Definition}[section]

\theoremstyle{proposition}
\newtheorem{proposition}{Proposition}[section]

\theoremstyle{corollary}
\newtheorem{corollary}{Corollary}[section]

\theoremstyle{lemma}
\newtheorem{lemma}{Lemma}[section]

\theoremstyle{theorem}
\newtheorem{remark}{Remark}[section]

\theoremstyle{theorem}
\newtheorem*{remark*}{Remark}

\theoremstyle{theorem}
\newtheorem{hypo}{Hypothesis}[section]

\theoremstyle{example}

\newcommand{\M}{\mathcal{M}}
\newcommand{\T}{\mathcal{T}}
\renewcommand{\d}{\partial}

\renewcommand{\H}{\mathcal{H}}
\newcommand{\A}{\mathcal{A}}
\newcommand{\B}{\mathcal{B}}
\newcommand{\I}{\mathrm{I}}
\newcommand{\slD}{\slashed{D}{}}
\newcommand{\slDs}{\slashed{D}_{\mathcal{S}^2}}
\newcommand{\s}{\mathcal{S}}
\newcommand{\R}{\mathbb{R}}
\newcommand{\C}{\mathbb{C}}
\newcommand{\g}{\mathbf{g}}
\newcommand{\f}{\tilde{f}}

\renewcommand{\S}{\mathbb{S}}
\newcommand{\diff}{\,\mathrm{d}}

\newcommand{\diag}{\mathrm{diag}}

\newcommand{\inner}[2]{\langle{#1},{#2}\rangle}
\newcommand{\Hsp}[1]{{H^{#1}_{\operatorname{sp}{}}}}
\newcommand{\Bm}{B^{-1/2}}
\newcommand{\Bnorm}[2]{\norm{\Bm #1 \Bm #2}}

\newcommand{\milbracket}[1]{\Bigg[\hspace{-4 pt}\Bigg[ #1 \Bigg]\hspace{-4 pt}\Bigg]}
\newcommand{\nablas}{\tilde{\nabla}}
\newcommand{\Q}{\mathcal{Q}}
\newcommand{\slcalD}{\slashed{\mathcal{D}}}

\begin{document}
	
\maketitle

\begin{abstract}
	In this paper we construct a scattering theory for the massive and charged Dirac fields in the interiors of sub-extremal Kerr-Newman(-anti)-de Sitter black holes. More precisely, we show existence, uniqueness and asymptotic completeness of scattering data for such Dirac fields from the event horizon of the black hole to the Cauchy horizon. Our approach relies on constructing the wave operators where the Hamiltonian of the full dynamics is time-dependent. To prove asymptotic completeness, we use two methods. The first involves a comparison operator, while for the second we introduce and employ a symmetry operator of the Dirac equation. 
\end{abstract}


\vspace{0.1in}


\tableofcontents


\section{Introduction}\label{sec:introduction}


 The scattering approach made its debut in general relativity decades ago and has now become a standard tool in probing many phenomena of black hole physics as well as in exploring their mathematics. From superradiance and Hawking radiation to stability of black hole exteriors,  scattering in the exterior regions of black hole  spacetimes was and still is being used extensively, and reviewing its literature in any sufficient degree of coverage is beyond the scope of this work. We therefore provide a sample which is by no means an exhaustive list for works on scattering outside black holes but which the reader can find a complete overview of the subject  therein, in addition to more references  \cite{daude_direct_2017,daude_sur_2004,joudioux_conformal_2012,bachelot_asymptotic_1994,bachelot_gravitational_1991,bachelot_hawking_1999,bachelot_scattering_1992,bachelot_scattering_1997,dafermos_scattering_2018,dimock_classical_1986,dimock_scattering_1985,dimock_scattering_1986,hafner_scattering_2004,hafner_sur_2003,bony_scattering_2005,georgescu_asymptotic_2017,pham_peeling_2017,besset_scattering_2021,nicolas_conformal_2016,nicolas_scattering_1995,mason_conformal_2004,mokdad_reissner-nordstrom-sitter_2017,mokdad_maxwell_2016}.    

In contrast, the list of mathematical studies of scattering theories in the interior regions of black holes is shy and young. Although the original motivation for these studies goes back to the 1970s due to a work by Penrose and Simpson \cite{simpson_internal_1973}, it is only recently that the subject has experienced a boom in the interest of the mathematics and physics communities of General Relativity. Nevertheless, this interest is growing very rapidly and many groups and researchers are actively working on problems inside black holes.

Shortly after formulating his singularity theorem with Hawking, Penrose put forth the cosmic censorship conjecture in 1969 in an attempt to save determinism in General Relativity. He suggested that the Cauchy horizon should be unstable and proposed the mechanism of blue-shift at the horizon as a cause for the instability. And in the aforementioned work of himself and Simpson, they ran numerical simulations in the interior of a Reissner-Nordström black hole and noticed a divergence in the energy flux of electromagnetic radiations measured by an observer crossing the Cauchy horizon. Following these observations, first McNamara \cite{mcnamara_instability_1997} and then  Chandrasekhar and Hartle \cite{chandrasekhar_crossing_1982}, used stationary scattering to study the blue-shift instabilities and the $C^1$-blow-up at the Cauchy horizon for gravitational perturbations. These perturbations are governed by wave equations, namely the Regge-Wheeler-Zerilli equations of polar and axial perturbations. Motivated by these works, in 2017 Dafermos and Shlapentokh-Rothman \cite{dafermos_time-translation_2015} provided a treatment of blue-shift instabilities for the scalar wave equation on a Kerr black hole using their earlier work of scattering in the exterior region \cite{dafermos_scattering_2018}. Kehle and Shlapentokh-Rothman in 2019 studied scattering of the geometric wave equation inside Reissner-Nordström black holes \cite{kehle_scattering_2019}. In their work, they used stationary scattering and constructed a complete scattering theory for linear waves in the interior of a Reissner-Nordström black hole. In addition, they proved negative results, i.e.,  a breakdown of scattering, in some cosmological and massive cases. The scattering theory they obtained was constructed directly from the outer horizon to the inner one, without concatenating  scattering maps from an intermediate Cauchy hypersurface to the horizons. In 2022, Mokdad and Nasser showed that in fact the usual construction of scattering using intermediate operators fails \cite{mokdad_scattering_2022}. This breakdown of intermediate scattering happens in all spherically symmetric black holes which possess a dynamic interior similar to that of a sub-extremal Reissner-Nordström black hole. In \cite{mokdad_scattering_2022}, it was shown that the obstruction to the classical construction of scattering lies in the behaviour of solutions at the zero spatial frequency and high angular momentum (spherical harmonic mode). In addition, the analytic framework used there shows that the phenomenon is generic to systems with exponentially decaying potentials, and not limited to the black hole context. Also in 2022 Luk, Oh, and Shlapentokh-Rothman gave another proof of the linear instability of the Reissner-Nordström Cauchy horizon by studying the linear scalar fields and analysing the scattering map at the zero frequency, providing the link between the blue-shift instability and the breakdown of intermediate scattering found earlier in \cite{mokdad_scattering_2022}. The same year, Sbierski provided a mathematical proof of the blue-shift instability at the Cauchy horizon in a sub-extremal Kerr black hole for the Teukolsky equation of spin-2 fields, also using a scattering approach \cite{sbierski_instability_2022}.

For Dirac fields, the first study was done by Häfner, Mokdad and Nicolas in 2021 \cite{hafner_scattering_2021}. In this work, a complete scattering theory for charged and massive Dirac fields was established in the interior of spherically symmetric black holes generalising a sub-extremal Reissner-Nordström black hole. Unlike the scalar field, the Dirac field naturally determines a current which is divergence-free, regardless of the spacetime. Should the spacetime be globally hyperbolic, we have a well-posed initial value problem with the current flux across a hypersurface of constant time being the norm, and the vanishing divergence translates to its conservation across this spacelike foliation. This enables the construction of the scattering map as a composition of two intermediate scattering maps defined on a spacelike hypersurface at a finite time. In \cite{hafner_scattering_2021}, the scattering was first achieved using the wave operators and then reinterpreted geometrically as the (inverse of the) trace operators. Shortly after, Mokdad obtained the same geometric scattering theory using the conformal approach \cite{mokdad_conformal_2022}. The latter did not rely on \cite{hafner_scattering_2021} and the Goursat problem was directly solved using a method  that transfers the setup to systems of waves followed by a reinterpretation of the solution as a Dirac field, see also \cite{mokdad_conformal_2019}. As a natural continuation of these works, in the current paper, we obtain a complete scattering theory in the language of wave operators for charged and massive Dirac fields inside sub-extremal Kerr-Newman-de Sitter and Kerr-Newman-Anti-de Sitter black holes, as well as all their admissible subfamilies, and in fact, in slightly more general spacetimes (see Section \ref{sec:configurations}).

In the above mentioned work \cite{hafner_scattering_2021} on Dirac fields, the spherical symmetry was used to simplify the arguments and reduce the problem to $1+1$-dimensions, and asymptotic completeness follows from Cook's method. This symmetry allows in particular an immediate decomposition of the hamiltonian on the spin-weighted spherical harmonics since the only angular part of the hamiltonian in that case is the Dirac operator on the sphere acting on Dirac bispinors: $\slD$ (defined in \eqref{eq:dirac op on s2}). The operator $\slD$ itself decomposes in a nice manner (see \eqref{eq:decomposion slD}). This decomposition is, however, not indispensable, and in fact was not used in the second work \cite{mokdad_conformal_2022}. In rotating black holes, this symmetry is lost and decomposing the hamiltonian becomes considerably more involved, and the reduction to 1+1-dimensions is not as straightforward. Nonetheless, it is still possible to decompose the hamiltonian on stable subspaces. This was done to establish scattering theories in the exterior of Kerr-type black holes \cite{daude_sur_2004,daude_direct_2017, hafner_scattering_2004, borthwick_scattering_2022}, see also \cite{belgiorno_absence_2009} for such a separation of variables but not necessarily for scattering purposes. The main difference between the exterior and the interior regions is that the hamiltonian becomes time-dependent, and this renders some decomposition methods that were used for the exterior, not so convenient for the interior of the black hole. For example, in \cite{daude_direct_2017}, the authors following \cite{belgiorno_absence_2009} decompose the full angular part of the hamiltonian, which is now more complicated than the $\slD$ operator and involves other angular derivatives due to the rotation of the black hole. The decomposition is done using the eigenvalues of the angular operator, however, these eigenvalues depend on the frequency, i.e., on the Fourier dual variable of the time variable. The resulting scattering theory is then at a fixed energy/frequency. In the interior, what was the time variable outside the black hole becomes a space variable, albeit still being the only natural variable for taking the Fourier transform. It is therefore not clear if such a decomposition would be useful in the interior. 

In principle, a suitable decomposition for the interior is probably still possible. However, it turns out that by using different techniques that are adapted to time-dependent hamiltonians and their dynamics which are $2-$parameter evolution systems, one can directly obtain useful  estimates on the evolution of the field and its derivatives. Here, we borrow techniques and tools from the theory of regularly generated dynamics which allow us to treat the problematic terms coming from the lack of symmetry in an adequate manner and without decomposing. As one expects, the Dirac operator on the sphere is the part that requires the most effort and treating it constitutes the main difficulty in this work. As in the spherically symmetric case in the work \cite{hafner_scattering_2021}, thanks to the exponential decay in time of the hamiltonian, no propagation estimates are needed to compare with the free dynamics. And in Cook's method, most of the terms in the difference between the full and the simplified hamiltonians are controlled in a straightforward manner by the initial data and therefore, it is clear that they decay exponentially fast. The $\slD$ term is however much more complicated since no simple commutation relation holds between it and the dynamics. To overcome this difficulty, we found two independent methods, each has its own interest.  

The first method relies on a comparison operator $B$ whose role is to help proving the necessary control in time on the $\slD$ derivatives of the Dirac field solution, this is the estimate in \Cref{prop:slD u bounded} of the paper. A crucial property of this comparison operator is that it generates a norm on the scattering states that is equivalent to the first Sobolev norm. Ultimately, this comes from a similar but fundamental property of $\slDs$, the Dirac operator acting on \textit{spinors} on the sphere. Namely, it defines a norm equivalent to the spinorial Sobolev norm. Then, the main idea of the proof of \Cref{prop:slD u bounded} is in essence an application of Grönwall's lemma where the exponent is essentially an operator norm involving the commutator of the full hamiltonian and the comparison operator $B$. Estimating this operator norm constitutes the majority of the calculations. 

From another side, in the second method, we use a symmetry operator for the Dirac equation. Similarly to the Carter operator which commutes with the scalar wave equation, there exits a second order symmetry operator $\mathcal{Q}$ for our massive Dirac equation in the general KN(A)dS which commutes with the Hamiltonian and it is remarkably well-suited to the form of the equations used here. In fact, we use this operator to provide a simplified proof which avoids the use of the operator $B$. The simplification occurs in the proof via Cook's method and bypasses the use of \Cref{prop:slD u bounded}. Although the proof using the symmetry operator is simpler, we do not expect it to be robust under perturbations of the fields, while we expect the proof using \Cref{prop:slD u bounded} to hold the same and to be more apt for studying perturbations. 

In this paper we provide both proofs, and even though we only use this symmetry operator to prove scattering without discussing it further, we expect it to have important applications in other contexts. To our knowledge, this is the first time the symmetry operator $\mathcal{Q}$ we introduce here appears in the literature, and may have been unknown before. The closest work we found that discusses a symmetry operator for Dirac fields, is the work by S. Jacobsson and T. Bäckdahl \cite{jacobsson_second_2023}. However, it seems that the only way to obtaining a symmetry operator from their framework is still by trial and error, and a somewhat elaborate computer algebra calculations, and therefore it does not simply follow from there work. We plan on discussing in more depth these symmetry operators in potential future collaboration.

As in previous works, it is interesting to reinterpret the analytic results geometrically. We postpone the geometric interpretation of our results to a future work since the geometry is considerably more complicated than the spherical case. One of the difficulties, even in the spherically symmetric case, concerning the geometric interpretation of the wave operators as the inverses of the trace operators is due to the charge, which require a change in the choice of the gauge of the ambient electromagnetic potential, so that it suits each horizon. In addition to the geometric interpretation, we also plan on studying the transferability of the regularity of a Dirac field from one horizon to the other in the non-Killing directions, in particular, in the transversal direction to the Cauchy horizon. From the scattering theory we construct here, it is readily seen that the regularity in the Killing directions is  transferred. It is worth mentioning that there is currently a similar work that will appear soon \cite{mokdad_Instabilities_2024} studying the regularity of Dirac fields in the transversal direction at the Cauchy horizon in the spherically symmetric case.  Regarding the spacetimes we consider here, we have excluded some geometries from our analysis, namely, the extreme black holes and part of the Anti-de Sitter subfamily. Both of these exclusions are for technical reasons related to the methods we use and, a priori, are not absolute obstacles for scattering. The extreme cases have double horizons and the decay there is not sufficient for our arguments. While some negative values of the cosmological constant create new coordinate singularities in the metric and some terms in our calculations become unbounded, consequently, different arguments are needed to control them. Thus, as part of the future plan, we would like to extend our result to all admissible values of negative cosmological constant and to the extreme black holes. Indeed, it would be interesting to see how the change in the geometry effects the scattering results, and not only the method. 

Modulo this restriction, this current work finishes the construction of the basic scattering theory for Dirac fields in the interior of classical black holes. Notwithstanding, other more exotic black hole spacetimes (e.g. a Vaidya-type metric) with suitable interior structure can be interesting to look at scattering in their interiors. On the other hand, and up to our knowledge, no scattering theories have been dedicated to the Maxwell fields, which we expect to be more involved than Dirac fields, since electromagnetic fields share many properties with the scalar fields, e.g. the absence of a positive definite and conserved quantity in the interior region.

The paper is organised as follows: In the next section we collect the used notations. In \Cref{sec:geometric setting} we present the metric and then we define the interior of a sub-extremal black hole in \Cref{subsec:interior}. There, we discuss the precise conditions on the free parameters of the metric in order to have a dynamic interior bounded by two horizons in each of the well-know subfamilies of the exact solutions to Einstein's equations. Afterwards we state the hypothesis we require on the metric components. In \Cref{subsec:asymptotics} we introduce a Regge-Wheeler-type variable and give the asymptotics of the relevant horizon function with respect to this variable. We finish the geometric setup with \Cref{subsec:PND} where we provide an adapted tetrad to our geometry and foliation.

\Cref{sec:Diracfield} introduces the Dirac fields starting with the bispinor bundle and the spinor components in the abstract index notation. The Dirac equation is in \Cref{subsec:Diracequation} and there we discuss the Dirac current and its conservation. In \Cref{subsec:NP-formalism} we recall the Newman-Penrose formalism and we give the projection of the Dirac equation on the spin-frame associated to our choice of normalised null tetrad. Next, we reformulate the problem as a Schrödinger equation with a hamiltonian. Before proceeding to scattering theory, we collect important properties and results for the Dirac operator $\slDs$ on the sphere in \Cref{sec:slDs}.

The scattering theory itself and the main results are in \Cref{sec:scattering}. We discuss the full and the simplified dynamics, after which we state the technical result of \Cref{prop:slD u bounded}. Then we state and prove the existence of the wave operators and their inverses as strong limits, as well as their unitarity.

\Cref{sec:B-operator} is dedicated to the proof  of \Cref{prop:slD u bounded}. It consists of the properties of the comparison operator, the commutator calculations and estimates. 

Finally, in the first appendix, we collect and prove the conditions on the parameters of the black hole to have a dynamic interior with two horizons. In the second appendix, we provide the main calculations for the spin coefficients used in the Newman-Penrose formalism in \Cref{subsec:NP-formalism}. 
 
\subsection*{Acknowledgment} The authors are deeply grateful to Thierry Daudé for his help and support on several occasions and for his valuable input to this work. We would like also to express our special thanks for Dietrich Häfner for his insightful observation regarding the symmetry operator which allowed a simplified proof of the main result. M. Mokdad acknowledges the appreciated support form the London Mathematical Society through the Atiyah-UK fellowship.

\subsection*{Notations and conventions}
We summarise here some of the notations and conventions used in this paper.
\begin{itemize}
	\item  For two functions $f$ and $g$, we write $f\lesssim g$ if there exists a constant $C>0$ such that $f(x)\le C g(x)$ for all $x$. We say that $f$ and $g$ are  equivalent, denoted by $f\eqsim g$, if $f\lesssim g$ and $g \lesssim f$. Also, we write $f\sim g$ if $f\eqsim g$ holds asymptotically, i.e., for all $x\ge x_0$. In particular, this implies 
	\begin{equation*}
	\lim\limits_{x\to\infty} \frac{f(x)}{g(x)} = C,
	\end{equation*} for some constant $C>0$. The hidden constant in $\lesssim$ may change from line to line in the calculations. 
	
	\item We denote by $\mathcal{C}_c^\infty(\mathcal{O})$ the collection of smooth compactly supported functions on a manifold $\mathcal{O}$ with values in $\mathbb{C}$, and by ${\cal C}^\infty_c ({\cal O}; F)$ the space of smooth sections of a fibre bundle $F$ which are compactly supported in ${\cal O}$. We use $\Gamma({\cal O}; F)$ to denote the space of all sections of $F$ over $\cal O$.
	\item In many parts of the paper, we rely on the abstract index formalism and notation for both spinors and vectors, see \cite{penrose_spinors_1987}.
	\item For a generic variable $x$, we use the shortened notation for derivation $\d_x=\frac{\d}{\d x}$, $D_x$ denotes $-i\d_x$, as well as  the abstract index notation $l^a\d_a$ for a generic vector field $l$.
	\item For Lorentzian metrics, we adopt the signature convention $(+\, -\, - \,\, - )$.
	\item  $\s^2$ is the unit 2-sphere with $\omega=(\theta,\varphi)\in  (0, \pi)_\theta \times (0, 2\pi)_\varphi$, 
\end{itemize}

\section{Geometric setting}\label{sec:geometric setting}

This section serves to present our geometric framework.

\subsection{The Kerr-Newman(-anti)-de Sitter metrics}\label{subsec:Kerr-metric}
We shall be considering a spacetime given by a 4-manifold $\mathcal{K}$ endowed with a Lorentzian metric $\g$ from the Kerr-Newman(-anti)-de Sitter (KN(A)dS) family of metrics, describing an eternal black hole, possibly with rotation and charge, in a spacetime with a cosmological constant that may be zero -- in which case we get the Kerr-Newman family of metrics. The suffixes de Sitter and anti-de Sitter refer to the cases when the cosmological constant is respectively positive and negative. 

The spacetime $(\mathcal{K},\g)$ is thus an exact solution of the Einstein-Maxwell coupled equations, and in Boyer-Lindquist coordinates, is given by $\mathcal{K} = \mathbb{R}_t \times \mathbb{R}^+_r \times \s^2_\omega$ and the metric $\g$ has the form:
\begin{multline}\label{eq:metric 0}
	\g = \frac{\Delta_r - \Delta_\theta a^2 \sin^2\theta}{\lambda^2 \rho^2}\diff t^2 + 	\left(\Delta_r a^2 \sin^2\theta - \Delta_\theta(r^2+a^2)^2 \right) \frac{\sin^2\theta}{\lambda^2\rho^2}\diff \varphi^2\\
-
	\frac{\rho^2}{\Delta_r} \diff r^2 -
	\frac{\rho^2}{\Delta_\theta} \diff \theta^2 +
	\left(\Delta_\theta (r^2+a^2) - \Delta_r \right)\frac{2a\sin^2\theta}{\lambda^2\rho^2} \diff t \diff \varphi ,
\end{multline}
or more compactly,
\begin{equation}\label{eq:metric 1}
	\g = \frac{\Delta_r}{\lambda^2 \rho^2}  \alpha \otimes \alpha - \frac{\rho^2}{\Delta_r}  \diff r^2 - \frac{\rho^2}{\Delta_\theta}\diff\theta^2 - \frac{\Delta_\theta \sin^2\theta}{\lambda^2\rho^2} \beta \otimes \beta,
\end{equation}
where $M>0$, $Q\in\mathbb{R}$; and we restrict  $a\in\mathbb{R}$, $\Lambda \in \R$ such that, \begin{equation}\label{eq:lor cond}
	\Lambda a^2 > -3 \, ,
\end{equation}
meaning\footnote{The restriction on $(\Lambda, a)$ imposed by \eqref{eq:lor cond} ensures that $\g$ as a Lorentzian metric is smooth in $\theta$ and well-defined, which is evident from the two positive quantities in \eqref{Delta_Theta-and-lembda}, respectively.}  
\begin{equation}\label{Delta_Theta-and-lembda}
 \Delta_\theta(\theta) := 1 + \frac{\Lambda a^2}{3} \cos^2{\theta} > 0, \quad \lambda := 1 + \frac{\Lambda a^2}{3} > 0;
\end{equation}
in addition, we have the horizon function
\begin{equation}\label{Deltar}
	 \Delta_r(r) := (r^2 + a^2)\bigg(1 - \frac{\Lambda r^2}{3}\bigg) - 2Mr + \lambda^2 Q^2;
\end{equation}
and 
\begin{equation}\label{rho}
\rho(r, \theta):= \sqrt{r^2 + a^2 \cos^2\theta};
\end{equation}
and finally, we have the 1-forms
\begin{align}\label{metric-1-forms}
\alpha_a \diff x^a := \diff{t}-a \sin^2{\theta} \, \diff \varphi,\quad \beta_a \diff x^a := a\, \diff{t}-(r^2+a^2)\diff{\varphi}.
\end{align}

As per usual notation for the Kerr-Newman parameters, $M$, $Q$ and $a$ are respectively the mass, the charge and the angular momentum per unit mass of the black hole. The de Sitter and anti-de Sitter aspects come from the quantity $\Lambda$, which is the cosmological constant. We refer to the quadruplet $(M,Q,a,\Lambda)$ as the free parameters of the KN(A)dS family of black holes. 

A priori, the metric \eqref{eq:metric 1} is to be understood as being defined for values of $r$ different from the roots of $\Delta_r$, yet, the metric can be regularly extended to cover the relevant roots of $\Delta_r$.  These roots correspond to important structures in the spacetime, namely the horizons, which are the null hypersurfaces $\{r=r_z\}$ for $r_z$ a positive root of $\Delta_r$. However, in this present work, we do not present the construction of the extensions of the spacetime $(\mathcal{K},\g)$ to include these roots since we perform no explicit calculations or local analysis there. For readers interested in these extensions, see e.g. \cite{borthwick_maximal_2018,hawking_large_1973,chandrasekhar_mathematical_1998,wald_general_2010,oneill_geometry_2014} for a general overview on the maximal analytic extensions. Also see \cite{hafner_scattering_2021,kehle_scattering_2019,sbierski_instability_2022} for brief descriptions of the constructions relevant to the interior of black holes. 

{Nonetheless, we shall see the horizons as asymptotic regions whose properties are essential for the dynamical approach of scattering used in this paper -- also known as time-dependent scattering. }


For the reader's convenience, we give the expression of $\det \g = -{\lambda^{-4}}{\rho^4 \sin^2\theta}$, from which follows the coordinate expression of the volume form of the metric $\g$: 
\begin{align*}
	\mathrm{dVol}_\g = \frac{\rho^2 \sin\theta}{\lambda^2} \diff t \wedge \diff r \wedge \diff \theta \wedge \diff \varphi,
\end{align*}
which will be used later. 
To fix an orientation on the manifold $\mathcal K$, we take $\mathrm{dVol}_\g$ to be positively oriented, or equivalently, by taking the chart $(t,r,\omega)$ to be positively oriented.  

For more on the KN(A)dS families, the reader can consult the classics \cite{wald_general_2010,hawking_large_1973,chandrasekhar_mathematical_1998} or the discussions in \cite{hintz_global_2018,borthwick_maximal_2018,podolsky_accelerating_2006}.
\subsection{The interior of a KN(A)dS black hole}\label{subsec:interior}

The interior region of the black hole we discuss in our work is distinguished by its horizon boundaries. For this, we give the conditions on the free parameters of the KN(A)dS metric for which these relevant horizons exist. Indeed, the number of horizons in the spacetime $(\mathcal K, \g)$ equals the number of positive real roots of the horizon function $\Delta_r$. We are interested in the case where the two smallest positive roots are simple and $\Delta_r$ is strictly negative in between these two roots. This ensures that our spacetime possesses an interior
bounded by an inner horizon of the black hole situated at the smallest root, called the Cauchy horizon; and an outer horizon at the larger root referred to as the black hole's event horizon.

\begin{remark}
The simplicity requirement on the multiplicity of the roots guarantees the asymptotics required for our method. Namely, a sufficiently fast decay of $\Delta_r$ as a function of the Regge-Wheeler variable, as discussed in section \ref{subsec:asymptotics}. This means that we  only treat the interior of a \underline{sub-extremal} KN(A)dS black hole.
\end{remark} 
Thus,
\begin{hypo}\label{hypo:1}
	We adopt the assumption that there exist $ r_\pm \in (0, +\infty) $ with $r_-<r_+$ such that:
		\begin{itemize}
			\item[(h1)] $ \Delta_r < 0 $ on $ (r_-, r_+) $,
			\item[(h2)] $ \Delta_r(r_\pm) = 0 \ne \Delta_r'(r_\pm) $.
	\end{itemize}
\end{hypo}
In view of these assumptions,
\begin{definition}\label{def:hypo1}
We define the \underline{interior of a subextremal KN(A)dS black hole} by the couple
\begin{equation}\label{eq:interior def}
(\mathcal M :=  \mathbb{R}_t \times (r_-, r_+)_r \times \s^2_\omega,\,\g{|_\mathcal{M}}),
\end{equation} viewed as a Lorentzian submanifold of the spacetime $(\mathcal K, \g)$, assuming \Cref{hypo:1}.
\end{definition}

\subsubsection{Configurations in subfamilies}\label{sec:configurations}

In what follows, we collect the conditions required on  the free parameters $(M,Q,a,\Lambda)$ of the different subfamilies of the KN(A)dS spacetimes so that \Cref{hypo:1} holds.

\begin{enumerate}[label=(\roman*)]
	
	\item \textbf{Non-black hole spacetimes ($M=0$):} Under {\textit{no}} configuration of the free parameters does \Cref{hypo:1} hold if $M=0$ since $\Delta_r$ will be positive between the two positive real roots when they exist (note that in this case, roots come in pairs of opposite sign). This immediately excludes spacetimes like \textit{Minkowski} $(M=0,Q=0,a=0,\Lambda=0)$, \textit{de Sitter} $(M=0,Q=0,a=0,\Lambda>0)$, or\textit{ anti-de Sitter} $(M=0,Q=0,a=0,\Lambda<0)$ spacetimes, to no surprise.
	
	\item \textbf{Exotic matter black holes} ($M<0$): In the case of negative mass, i.e. $M<0$, again no conditions on the parameters are compatible with \Cref{hypo:1}. This is proven in  \Cref{appendix:Config}.
	
	 {\textit{From now on we thereby assume $M>0$}}.

	\item \textbf{Non-cosmological black holes ($\Lambda=0$):} The horizon function associated to the \textit{Kerr-Newman} metric $(M>0,Q\ne0,a\ne 0,\Lambda=0)$, including its special cases, the \textit{Kerr} $(M>0,Q=0,a\ne 0,\Lambda=0)$ and the \textit{Reissner-Nordström} $(M>0,Q\ne0,a=0,\Lambda=0)$ metrics, reduces to $\Delta_r(r)=r^2-2Mr+a^2+Q^2$ and has two simple positive real roots if and only if $M^2>a^2+Q^2>0$. The roots $r_\pm$ for which \Cref{hypo:1} holds are:
	\begin{equation*}
	r_\pm=M\pm \sqrt{M^2-a^2-Q^2}.
	\end{equation*}
	Clearly, the condition $a^2+Q^2>0$ excludes the \textit{Schwarzschild} metric $(M>0,Q=0,a=0,\Lambda=0)$, as expected.  
	\item \textbf{Reissner-Nordström-type black holes ($a=0$):} For the \textit{Reissner-Nordström-de Sitter} $(M>0,Q\ne0,a= 0,\Lambda>0)$ metric, the different configurations of horizons were fully analysed in \cite{mokdad_reissner-nordstrom-sitter_2017}. While for the \textit{Reissner-Nordström-anti-de Sitter} case $(M>0,Q\ne0,a=0,\Lambda<0)$, it was analysed in \cite{hafner_scattering_2021}. Note that in both cases, the analyses in \cite{mokdad_reissner-nordstrom-sitter_2017,hafner_scattering_2021} show that a non-zero charge must be present in order for the \Cref{hypo:1} to hold, i.e., $Q\ne 0$ is a necessary condition for \Cref{hypo:1}. As one expects, this also excludes the de Sitter and anti-de Sitter versions of the Schwarzschild metric $(M>0,Q=0,a=0,\Lambda\gtrless 0)$. 
	
{	Since the scattering of massive and charged Dirac fields in the interior of Reissner-Nordström-type black holes was completely treated in \cite{hafner_scattering_2021}, we do not give here the explicit conditions for  \Cref{hypo:1} to hold in this case, and we refer the reader to the mentioned work.  }  
	
	\item \textbf{Rotating-de Sitter black holes ($\Lambda >0$, $a\ne0$):} The analysis for the\textit{ Kerr-de Sitter} metric  $(M>0,Q=0,a\ne0,\Lambda>0)$ was done by J. Borthwick in \cite{borthwick_maximal_2018}. The precise conditions on the parameters of a \textit{Kerr-Newman-de Sitter} metric $(M>0,Q\ne0,a\ne0,\Lambda>0)$ are analysed in the appendix. For both of these cases, we give the explicit conditions for which \Cref{hypo:1} holds, this is  \Cref{prop:distinct horizons} and its remark. 
	
	When $\Lambda > 0$, a third horizon may also be present in the spacetime $(\mathcal K, \g)$, called the cosmological horizon, at a larger root of $\Delta_r$. Evidently, this horizon is not included in the interior and turns out to have no effect on the scattering of Dirac fields in the interior region\footnote{In other cases, like the scattering of linear waves (see \cite{kehle_scattering_2019}), the cosmological constant $\Lambda$ has more influence on the existence of the scattering maps.}. That is, our construction of the scattering theory holds equally well whether $\Lambda=0$ or $\Lambda>0$.
	
	\item\textbf{Rotating-anti de Sitter black holes ($\Lambda <0$, $a\ne0$):} The admissible cases and conditions are given in the appendix. Note that, in the body of the paper, $\Lambda a^2>-3$ is assumed in order to retain Lorentzian signature.

\end{enumerate}

\subsubsection{Hypotheses on the metric}

All of the desirable cases in the above survey can be summarised using the following remark. 

\begin{remark}\label{rk:general horizon fn}
Although the main interest lies in the KN(A)dS form of the horizon function $\Delta_r$ that  appears in (\ref{Deltar}), our analysis applies to any other radial function (i.e., with no dependence on $t$, $\theta$ or $\varphi$) that satisfies  \Cref{hypo:1} and which is smooth on the compact interval $[r_-,r_+]\subset(0,+\infty)$ of the Hypothesis.	In fact, these are the only properties we use for the function $\Delta_r$.
\end{remark}
With the above remark in mind, we state the final assumptions on the components of the metric $\g$ needed for our approach to constructing the scattering theory. 

\begin{hypo}\label{hypo:2} 	For $ r_\pm \in (0, +\infty) $ with $r_-<r_+$, and the spacetime $(\mathcal M, \g)$ with 
	\begin{gather*}
			\mathcal M :=  \mathbb{R}_t \times (r_-, r_+)_r \times \s^2_\omega, \\
				\g := \frac{\Delta_r}{\lambda^2 \rho^2}  \alpha \otimes \alpha - \frac{\rho^2}{\Delta_r}  \diff r^2 - \frac{\rho^2}{\Delta_\theta}\diff\theta^2 - \frac{\Delta_\theta \sin^2\theta}{\lambda^2\rho^2} \beta \otimes \beta,
	\end{gather*}
where $\Delta_\theta$ and $\lambda$ are given as in (\ref{Delta_Theta-and-lembda}) with fixed  $(a,\Lambda)\in\R^2$; $\rho$ as in \eqref{rho}; $\alpha$ and $\beta$ as in (\ref{metric-1-forms}); we assume that
		\begin{enumerate}[label*=(H\arabic*)]
			\item $ \Delta_r\in\mathcal{C}^\infty([r_-r_+])$ is any smooth function\footnote{In fact, we only need $\Delta_r$ to be twice continuously differentiable, but we shall keep smoothness for consistence with previous works \cite{hafner_scattering_2021,mokdad_scattering_2022,mokdad_conformal_2022}.} of the variable $r$ only, \label{H1}
			\item $ \Delta_r < 0 $ on $ (r_-, r_+) $, \label{H2}
			\item $ \Delta_r(r_\pm) = 0 \ne \Delta_r'(r_\pm) $, \label{H3}

 	\item $\Lambda a^2>-3$. \label{H4}
	\end{enumerate}

\end{hypo}

{\begin{center}
		\underline{We admit \Cref{hypo:2} for the rest of the paper.}\end{center}}

Like $\mathcal K$, we orient $\mathcal M$ by declaring the chart $(t, r, \theta, \phi)$ as positively oriented. Moreover, we adopt the time orientation for which $-\partial_r$ is future oriented, which is indeed timelike on $\mathcal M$ thanks to \ref{H2}.


\subsection{The Regge-Wheeler coordinate}\label{subsec:asymptotics}

One way of dealing with the coordinate singularities at $r=r_\pm$ is to push them to infinity. This can be done using a Regge-Wheeler-type coordinate. This change of coordinates gives rise to simpler calculations and to an analytic framework adapted to the wave operators approach to scattering. Let
\begin{equation}\label{f}
	f(r):=\frac{\Delta_r(r)}{\lambda (r^2+a^2)}.
\end{equation} 
The Regge-Wheeler coordinate $\tau$ is then defined by requiring
\begin{align}\label{eq:reggewheeler}
	&\frac{\diff{\tau}}{\diff{r}}(r) = \frac{1}{f(r)},
\end{align}
and any initial condition $\tau(r_0)=0$ for some $r_0\in(r_-,r_+)$. Since $f<0$, $\tau$ is a strictly decreasing function of $r$ which ranges from $ -\infty $ to $ +\infty $ as $ r $ runs from $ r_+ $ to $ r_- $.

As $\d_r$ is timelike, $\d_\tau$ is also timelike, and $\tau$ will represent the time variable with respect to which we will view the evolution of Dirac fields in the interior region. On the other hand, $\d_t$ is spacelike as can be seen from (\ref{eq:metric 0}). Because of that, we relabel $t$ into $x$ to insinuate the spacelike nature of the coordinate variable $t$ in the interior. 

We now define the global chart $(\mathbb{R}_\tau \times \mathbb{R}_x \times \s^2_\omega, (t,x,\omega))$ on $\mathcal M$, and in practice, we will identify them with each other. 
In this chart, $(\tau, x, \theta, \varphi)$, the metric $\g$ in \eqref{eq:metric 1} becomes
\begin{equation}\label{eq:metric 2}
	\g =  \frac{f^2\rho^2}{-\Delta_r}\diff \tau^2 + \frac{\Delta_r}{\lambda^2 \rho^2}  \alpha \otimes \alpha - \frac{\rho^2}{\Delta_\theta}\diff\theta^2 - \frac{\Delta_\theta \sin^2\theta}{\lambda^2\rho^2} \beta \otimes \beta \, ,
\end{equation}	
where $r$ can now be seen as an implicit function of $\tau$, defined by \eqref{eq:reggewheeler}. Moreover, it is clear that $\partial_\tau$ has the same time orientation as $-\d_r$ and is thus future oriented.

Note the flip in the order of the variables $(\tau, x=t, \theta, \varphi)$ with respect to the original Boyer-Lindquist chart $(t=x, r, \theta, \varphi)$. This compensates for the negative sign of $f$ in (\ref{f}), and therefore, the new chart is still positively oriented on $\mathcal M$. This can also be seen from the expression of the metric 4-volume form in this chart:
\[\mathrm{dVol}_\g=\frac{-f\rho^2 \sin\theta}{\lambda^2}\diff \tau \wedge \diff x \wedge \diff\theta \wedge \diff\varphi.\]

The function $f$ contains the same information as $\Delta_r$, and in fact satisfies exactly the same hypotheses, namely, \ref{H1}, \ref{H2} and \ref{H3}, with the same $r_\pm$. The most important feature of $f(r(\tau))$ is its exponentially decaying asymptotics in terms of $\tau$: 
\begin{lemma}\label{lem:asymptotics}
	Let $ r_\pm \in (0, +\infty) $ with $r_-<r_+$, and let $f(r)$ be a real-valued function of one variable ($r$) satisfying\footnote{That is, in place of $\Delta_r$.} \ref{H1}--\ref{H3}. Assume moreover that \eqref{eq:reggewheeler} holds, and let $r(\tau)$ be the inverse function of $\tau(r)$. Finally, set 
	\begin{equation}\label{kappa_pm}
		\kappa_\pm := \frac{1}{2}\frac{\diff f }{\diff r}(r_\pm).
	\end{equation}
Then
\begin{enumerate}[label*=\roman*.]
	\item $\kappa_- < 0<\kappa_+$.
	\item  There exists a constant $\tau_L >0$, depending on the function $f$ only such that
	\begin{align*}
		&    r(\tau)-r_-\eqsim -f(r(\tau)) \eqsim e^{2\kappa_- \tau} \quad \forall \tau > \tau_L \geq 0 , \\
		&r_+-r(\tau)\eqsim -f(r(\tau)) \eqsim e^{2\kappa_+ \tau} \quad \forall \tau < -\tau_L \leq 0 .
	\end{align*}
i.e.,
\begin{equation}\label{eq:f asymptotics}
	\abs{r-r_\pm} \sim e^{2\kappa_\pm \tau} \quad \text{and} \quad \abs{f} \sim e^{2\kappa_\pm \tau} \qquad \text{as}~~ \tau\rightarrow\mp\infty. 
\end{equation}
\end{enumerate}
\end{lemma}
\begin{proof}
	This follows directly form \ref{H1}--\ref{H3} and \eqref{eq:reggewheeler}. The explicit calculations have been carried out in the previous works, \cite[Equation~(3)]{hafner_scattering_2021} and \cite[Lemma~7.iii]{mokdad_scattering_2022}.
\end{proof}

\begin{remark}\label{rem:asymp-Deltar}
	Clearly, \Cref{lem:asymptotics} applies to $\Delta_r$ in place of $f$. 
\end{remark}

\subsection{Adapted null tetrad}\label{subsec:PND}

To use the Newman-Penrose formalism, we need to construct a null tetrad adapted to our geometric framework. We rely on principal null directions of the spacetime. These arise from the Weyl tensor and define the Petrov type of the spacetime, which is of type D in the case of KN(A)dS (see e.g. \cite{oneill_geometry_2014} or \cite{chandrasekhar_mathematical_1998}). 

In Petrov type D, there are two double principal null directions, and these are given in our spacetime (see\footnote{See also \cite{daude_direct_2017} and the references therein.} \cite[Proposition~1]{borthwick_maximal_2018}) by the future oriented null vectors	

\begin{equation}\label{eq:principal null directions}
V^\pm := -\partial_r \mp \frac{1}{\lambda} \alpha^a \partial_a 
= -\frac1{f} \left(\partial_\tau \pm \left(\partial_x + \frac{a}{r^2+a^2}\partial_\varphi\right)\right).
\end{equation}

A null tetrad $\mathbb{T} = \{e_{i}\}_{i\in\{1,2,3,4\}} = \{l, n, m, \bar m\}$ is basis of the complexified tangent space at each point of a local chart of the spacetime, i.e., a local frame. It is said to be global if it forms a global frame. It consists of four null vectors, two of which are real, while the other two are complex. Here, the real vectors of the tetrad are $l$ and $n$, and are taken to be future oriented, while the complex vectors $m$ and $\bar{m}$ are complex conjugate of each other. We say that the tetrad $\mathbb{T}$ is normalised if its vectors satisfy the following normalisation conditions:

\[ l_a n^a = -m_a\bar{m}^a = 1 \, ;\]
while all the other products are zero.

We form our normalised global tetrad $\mathbb{T}$ on $\mathcal{M}$, adapted to the geometry of the KN(A)dS spacetime, with $l$ and $n$ obtained by normalising the null vectors $V^\pm$ given in \eqref{eq:principal null directions}:
\begin{equation}\label{eq:tetrad}
\mathbb{T} :
\begin{cases}
l^a \partial_a &= \frac{1}{-f}\sqrt{\frac{-\Delta_r}{2\rho^2}} 
\left(\partial_\tau + \partial_x + \frac{a}{r^2+a^2}\partial_\varphi\right)
,
\\[10 pt]
n^a \partial_a &= \frac{1}{-f}\sqrt{\frac{-\Delta_r}{2\rho^2}} 
\left(\partial_\tau - \partial_x - \frac{a}{r^2+a^2}\partial_\varphi\right)
,
\\[10 pt]
m^a \partial_a &= 
\sqrt{\frac{\Delta_\theta}{2 \rho^2}}
\bigg(
\partial_\theta + \frac{ia\lambda\sin\theta}{\Delta_\theta}\partial_x
+
\frac{i\lambda}{\Delta_\theta \sin\theta}\partial_\varphi
\bigg)
,
\\[10 pt]
\bar{m}^a \partial_a &= 
\sqrt{\frac{\Delta_\theta}{2 \rho^2}}
\bigg(
\partial_\theta- \frac{ia\lambda\sin\theta}{\Delta_\theta}\partial_x
-
\frac{i\lambda}{\Delta_\theta \sin\theta}\partial_\varphi
\bigg).
\\
\end{cases}
\end{equation}

Additionally and as we shall shortly see, this normalised Newman-Penrose tetrad is chosen to be adapted to the foliation of $\M$ defined by the hypersurfaces $$\Sigma_\tau:=\{\tau\}\times\R_x\times\s^2_\omega,$$ in the sense that ${l^a+n^a}$ points in the direction of $T^a$, the future-oriented unit normal to $\Sigma_\tau$. In fact, we have
\begin{equation} \label{eq:normal to sigma_t}
	\dfrac{l^a+n^a}{\sqrt{2}} = T^a \, , \quad T^a \partial_a = \frac{\sqrt{-\Delta_r}}{-f\rho}\, {\partial_\tau} \, .
\end{equation}

\section{The Dirac field}\label{sec:Diracfield}

In this section we recall the required material for spinors and the Dirac field with its current. We then give the hamiltonian formulation of the Dirac equation suitable for scattering in the interior of KN(A)dS black holes. For a general account on spinor, see for example any of  \cite{hijazi_spectral_2001,carroll_lecture_1997,abrikosov_jr_dirac_2002,chandrasekhar_mathematical_1998,penrose_spinors_1987}.

\subsection{Dirac bispinor bundle}\label{subsec:bispinorbundle}

Since $\mathcal{M}$ has a Cauchy hypersurface (e.g. $\Sigma_\tau$), it is globally hyperbolic. It follows that it admits a spin structure. The spin structure is given by the spinor bundle $\S^A$ and its complex conjugate bundle $\S^{A'}$, over $\mathcal{M}$.  \textit{Throughout the whole paper, the word ``spinor'' will always refer to a spin-1/2 spinor\footnote{Also known as 2-component spinor.}, i.e., an element of the aforementioned bundles (or their duals)}. 

The bundle $\S^A$ is equipped with a symplectic form $\varepsilon_{AB}$, and its complex conjugate $\varepsilon_{A'B'}$ acts on $\S^{A'}$. These symplectic structures allow one to canonically map $\S^A$ and $\S^{A'}$ to their dual spinor bundles, denoted by $\S_A$ and $\S_{A'}$, respectively.
Being a complex vector bundle, $\S^A$ admits a local basis around each point of $\mathcal{M}$.  Such a basis is called a spin dyad since the fibres of the bundle are of two complex dimensions, and it is customary to denote the spin dyad by $\{ o^A , \iota^A \}$. The complex conjugate of a spin dyad $\{ o^A , \iota^A \}$ is a spin dyad of the conjugate bundle, denoted by $\{ \bar{o}^{A'} , \bar{\iota}^{A'}\}$. And using the symplectic form, the spin dyad determines a spin dyad of the dual bundle $\{ o_B= \varepsilon_{AB} o^A, \iota_B=\varepsilon_{AB}\iota^A \}$, and of course it complex conjugate for $\S_{A'}$. When $o_A \iota^A =1$, we say that the spin dyad is normalised, and refer to it as a spin-frame. 

A spinor field $\phi_A$ is an element of $\Gamma (\S_A )$, the space of sections over $\mathcal M$ of the spinor bundle $\S_A$. The components of a spinor field in a spin-frame $\{ o^A , \iota^A\}$ are given as follows (see \cite[p.111-112]{penrose_spinors_1987} or \cite[Chapter~ 10]{chandrasekhar_mathematical_1998}): \\
For a spinor field $\phi_A \in \Gamma (\S_A )$,
\[ \phi_0 = \phi_A o^A \, ,~ \phi_1 = \phi_A \iota^A ,\]
and, say for $\chi^{A'} \in \Gamma (\S^{A'})$,
\[ \chi^{0'} = -\bar{\iota}_{A'} \chi^{A'} \, ,~ \chi^{1'} = \bar{o}_{A'} \chi^{A'} \, . \]
Because of $o_A\iota^A=1$ and the anti-symmetry of $\varepsilon_{AB}$, we have 
\begin{equation}
	\phi_A=\phi_1 o_A -\phi_0\iota_A \qquad \text{and} \qquad \chi^{A'}=\chi^{0'}\bar{o}^{A'} + \chi^{1'}\bar{\iota}^{A'}.
\end{equation}
Note the change in sign when raising and lowering the numerical index, e.g.,
$$\phi_0=-\phi^1 \qquad \text{and} \qquad \phi_1=\phi^0.$$

The spin structure is closely tied to the Lorentzian structure on spacetimes. There is a 1-to-2 correspondence between normalised null tetrads and spin-frames (see e.g., \cite{penrose_spinors_1987}). To any normalised null tetrad\footnote{Not necessarily the tetrad given in \eqref{eq:tetrad}.} $\mathbb{T}=\{l^a, n^a, m^a, \bar m^a\}$ there correspond uniquely two spin-frames $\{ \pm o^A , \pm \iota^A \}$ that differ by an overall sign only\footnote{This ambiguity in the sign reflects the double covering nature of spinors over the restricted Lorentz group \cite{penrose_spinors_1987}.}, such that
\begin{equation} \label{NPTSF}
l^{AA'}:=l^a = o^A \bar{o}^{A'} \, ,~ n^{AA'}:=n^a = \iota^A \bar{\iota}^{A'} \, ,~ m^a = o^A \bar{\iota}^{A'} \, ,~ \bar{m}^a = \iota^A \bar{o}^{A'} \, .
\end{equation}
This correspondence means that the spinor bundle combined with its conjugate can be identified with the complexified tangent bundle over $\mathcal M$:
\begin{equation}\label{complexifiedTS}
	\S^A \otimes \S^{A'}= T^a\mathcal{M} \otimes \C.
\end{equation}
Moreover, the symplectic form $\varepsilon_{AB}$ together with its conjugate $\varepsilon_{A'B'}$, can be used to decompose the metric: $$\g_{ab} = \varepsilon_{AB} \varepsilon_{A'B'}.$$
Note that due to the anti-symmetric properties, the product $\kappa^A\bar\kappa^{A'}$ of any spinor $\kappa^A$ with its complex conjugate $\bar\kappa^{A'}$ is a null vector.

Finally, the bundle of Dirac bispinors is given by $\S_A \oplus \S^{A'}$. We shall denote one of its elements by $(\phi, \chi)$ or $(\phi_A, \chi^{A'})$. A Dirac field is then a section of this bundle, i.e., an element of $\Gamma(\M;\S_A \oplus \S^{A'})$, and it is an example of a a spin-1/2 fermion field. When projected on a tetrad $\mathbb{T}$, that is, on an associated spin-frame $\{o^A,\iota^A\}$, this bispinor gives a complex 4-vector $(\phi_0, \phi_1, \chi^{0'}, \chi^{1'})$, and we have the following identities:
\begin{equation}\label{project-on-tetrad}
\begin{aligned}
	&\phi_A \bar{\phi}_{A'}l^{AA'}= \abs{\phi_0}^2,
	&\phi_A \bar{\phi}_{A'}n^{AA'}= \abs{\phi_1}^2,\\
	&\chi^{A'} \bar\chi^{A}l_{AA'}= |{\chi^{0'}}|^2,
	&\chi^{A'} \bar\chi^{A}n_{AA'}= |{\chi^{1'}}|^2.\\
\end{aligned}
\end{equation}

\subsection{The Dirac equation and the conserved current}\label{subsec:Diracequation}
Using the abstract index notation for the metric connection,  the charged and massive Dirac equation for the bispinor $(\phi, \chi)$ can be written as
\begin{equation}\label{eq:dirac}
\begin{cases}
(\nabla^{AA'}-iqA^{AA'})\phi_A &= \frac{m}{\sqrt2}\chi^{A'},
\\
(\nabla_{AA'}-iqA_{AA'})\chi^{A'} &= -\frac{m}{\sqrt2}\phi_A,
\end{cases}
\end{equation}
where $m$ and $q$ are respectively the mass and the charge of the bispinor field, and
\begin{equation}\label{eq:gauge potential}
A_a dx^a := \frac{Qr}{\rho^2} (\diff{x} - a \sin^2\theta \, \diff\varphi) = \frac{Qr}{\rho^2} \alpha.
\end{equation}
is a gauge potential of the ambient electromagnetic field of the KN(A)dS spacetime. Note that we allow massless and/or uncharged fields, i.e., the constants $m$ and $q$ can be zero. 

\begin{remark}
	We recall that, due to the gauge independence of the Maxwell equations, for any $\xi \in \mathcal{C}^\infty(\M)$, the gauge potentials $A$ and $A + \diff \xi$ yield the same solution $(\g, F)$ to the Einstein-Maxwell equations, with $\g$ the background metric \eqref{eq:metric 0} and $F = \diff A$ the Faraday tensor. Furthermore, the Dirac equation is invariant under the gauge transformation
	\begin{equation}\label{eq:gauge on A}
		A \mapsto \tilde{A} = A + \diff \xi \, , \quad (\phi_A, \chi^{A'}) \mapsto (\tilde\phi_A, \tilde\chi^{A'}) = e^{iq\xi} (\phi_A, \chi^{A'}) \, .
	\end{equation}
	Indeed, using $\nabla \xi = \diff \xi$, we have
	\begin{align}\label{eq:gt dirac}
		\left(\nabla^{AA'} - i q A^{AA'}\right) \phi_A &\mapsto \left(\nabla^{AA'} - i q A^{AA'} - i q \nabla^{AA'}\xi \right) e^{iq\xi} \phi_A \nonumber\\&= e^{iq\xi} \left(\nabla^{AA'} - i q A^{AA'}\right) \phi_A \, .
	\end{align}
	As $\abs{e^{iq\xi}}=1$, $(\tilde\phi_A, \tilde\chi^{A'})$ satisfy \eqref{eq:dirac} with $A$ replaced by $\tilde{A}$ if and only if $ (\phi_A, \chi^{A'})$ satisfy the original \eqref{eq:dirac}.
\end{remark}

The first important property of the Dirac equation is the conservation of  the Dirac current. For a Dirac bispinor $(\phi_A, \chi^{A'})$, its current is defined to be the vector field
\begin{equation}\label{eq:dirac current}
J^a = \phi^A \bar\phi^{A'} + \chi^{A'} \bar\chi^{A}.
\end{equation}
Clearly, the current is gauge independent, as can be seen from \eqref{eq:gauge on A}. Moreover, the current is in fact a future oriented causal vector field. This is because it is the sum of two null vector fields, each of which is future oriented as can be seen from \eqref{eq:normal to sigma_t} and \eqref{project-on-tetrad}.

\begin{lemma}\label{lem:current-div-free}
	When $(\phi_A, \chi^{A'})$ satisfies the Dirac equation \eqref{eq:dirac}, the current $J$ of \eqref{eq:dirac current} is divergence-free, i.e.,
\begin{equation}\label{eq:conserved current}
\nabla^a J_a = 0 \,.
\end{equation}
\end{lemma}
\begin{proof}
	See, e.g., \cite[Equation~(15)]{hafner_scattering_2021}.
\end{proof}
In view of \eqref{eq:conserved current}, it is natural to consider the flux of $J$ across a hypersurface $\Sigma_\tau$ 
$$\int_{\Sigma_\tau} J_a T^a \, \mathrm{dVol}_{\Sigma_t},$$
where $	\mathrm{dVol}_{\Sigma_\tau}$ is the induced volume form on $\Sigma_\tau$
\begin{equation}\label{dVolSigmat}
	\mathrm{dVol}_{\Sigma_\tau} :=i_T(\mathrm{dVol}_{\g})= \frac{\sqrt{-\Delta_r \, \rho^2}}{\lambda^2} \sin\theta \diff x \wedge\diff \theta\wedge\diff\varphi \, ,
\end{equation}
where $i_T$ is the interior product by the vector $T$. Or as a measure,
\begin{equation}\label{dVolSigmatMEASURE}
	\mathrm{dVol}_{\Sigma_\tau} = \frac{\sqrt{-\Delta_r \, \rho^2}}{\lambda^2} \diff x \diff \omega \,
\end{equation}
with $\diff \omega$ the Lebesgue measure on the Euclidean unit sphere $\s^2$.

Thanks to our choice of tetrad, the flux of $J^a$ gives a natural and simple expression of the $L^2$-norm of $(\phi, \chi)$ on $\Sigma_\tau$. Indeed, let $\{ o^A , \iota^A \}$ be the spin-frame related by \eqref{NPTSF} to the tetrad $\mathbb{T}$ in \eqref{eq:tetrad}, then by \eqref{eq:normal to sigma_t} and \eqref{project-on-tetrad}, we have
\begin{equation} \label{FluxSigmat}
\int_{\Sigma_t} J_a T^a \, \mathrm{dVol}_{\Sigma_\tau} = \frac{1}{\sqrt{2}} \int_{\Sigma_\tau} \left( \vert \phi_0 \vert^2 + \vert \phi_1 \vert^2 + \vert \chi^{0'} \vert^2 + \vert \chi^{1'} \vert^2 \right) \mathrm{dVol}_{\Sigma_t} \, .
\end{equation}

Setting $\Phi  := {}^t (\phi_0, \phi_1, \chi^{0'}, \chi^{1'})$, we define the space ${\cal H}_\tau$ for each $\tau \in \R$ to be
\begin{equation}\label{SpaceHt}
{\mathcal H}_\tau := L^2 (\Sigma_\tau;\S_A \oplus \S^{A'} ) \,  \quad\text{equipped with}\quad \Vert(\phi,\chi) \Vert^2_{{\mathcal H}_\tau} = \frac1{\sqrt2}\int_{\Sigma_\tau} \Vert{\Phi}(\tau)\Vert_{\C^4}^2 \, \mathrm{dVol}_{\Sigma_\tau} \, .
\end{equation}
With a slight abuse of notation, we sometimes write $\Vert\Phi \Vert_{{\mathcal H}_\tau}$ to mean $\Vert(\phi,\chi) \Vert_{{\mathcal H}_\tau}$.

Although they are well-known facts, for the sake of completeness, we now state the well-posedness of \Cref{eq:dirac} and the conservation of the current flux.

\begin{lemma}\label{lem:Well-posed-andCC}
	Fix $s\in\R$ and let $(\alpha_A,\beta^{A'})\in\H_s$, then there exists a unique Dirac field $(\phi_A,\chi^{A'})\in\mathcal{C}(\R_\tau;\H_\tau)$ solving \Cref{eq:dirac} in the weak sense, such that 
	\begin{equation*}
		(\phi_A,\chi^{A'})\vert_{\Sigma_s}=(\alpha_A,\beta^{A'}).
	\end{equation*}
and the norm is conserved in time: $\forall \tau\in\R$
\begin{equation}\label{eq:norm-conservation}
	\Vert(\phi_A,\chi^{A'}) \Vert_{{\mathcal H}_\tau}=\Vert(\alpha_A,\beta^{A'}) \Vert_{{\mathcal H}_s}. 	
\end{equation}
Moreover, any added regularity in the initial data is transferred to the solution at all times, in addition to continuity in time.
\end{lemma}
\begin{proof}
  Since our spacetime $\M$ is globally hyperbolic and $\Sigma_s$ is a Cauchy hypersurface, existence and uniqueness follow from the standard theory of hyperbolic equations. For example, by density, it is enough to show it for smooth compactly supported initial data $(\alpha_A,\beta^{A'})\in\mathcal C^\infty_c(\Sigma_s;\S_A \oplus \S^{A'} )$. The result then follows from Leray's Theorems \cite{leray_hyperbolic_1955}, or for a more recent reference, see \cite[Theorem~4]{andersson_wave_2018}, and the solution $(\phi_A,\chi^{A'})\in\mathcal{C}^\infty(\M;\S_A \oplus \S^{A'})$.
  
  Identity \eqref{eq:norm-conservation} can be  proven again by density. Using finite speed propagation \cite[Theorem~4]{andersson_wave_2018}, a simple application of the divergence theorem with  \eqref{eq:conserved current} gives \eqref{eq:norm-conservation} for initial data in $C^\infty_c(\Sigma_s;\S_A \oplus \S^{A'} )$.
\end{proof}

\subsection{Newman-Penrose formalism}\label{subsec:NP-formalism}

We now use the Newman-Penrose formalism to obtain a system of coupled PDEs on the spin components $\Phi  = {}^t (\phi_0, \phi_1, \chi^{0'}, \chi^{1'})$ of the bispinor $(\phi_A,\chi^{A'})$. This system of PDEs is \Cref{eq:dirac} in the Newman-Penrose formalism, and it takes the following form (see for instance \cite[Section~103]{chandrasekhar_mathematical_1998}): 
\begin{equation}\label{eq:dirac eq NPF}
\begin{aligned}
	&
	n^a(\partial_a - iqA_a) \phi_0 - m^a(\partial_a - iqA_a) \phi_1 + (\mu_s-\gamma_s) \phi_0 + (\tau_s-\beta_s) \phi_1 = \frac m {\sqrt2} \chi^{0'}
	\,,
	\\
	&
	l^a(\partial_a - iqA_a) \phi_1 - \overline{m}^a(\partial_a - iqA_a) \phi_0 + (\alpha_s-\pi_s) \phi_0 + (\epsilon_s-\rho_s) \phi_1 = \frac m {\sqrt2} \chi^{1'}
	\,,
	\\
	&
	l^a(\partial_a - iqA_a) \chi^{0'} + m^a(\partial_a - iqA_a) \chi^{1'} + \overline{(\epsilon_s-\rho_s)} \chi^{0'} - \overline{(\alpha_s-\pi_s)} \chi^{1'} = -\frac m {\sqrt2} \phi_0 
	\,,\\
	&
	n^a(\partial_a - iqA_a) \chi^{1'} + \overline{m}^a(\partial_a - iqA_a) \chi^{0'} - \overline{(\tau_s-\beta_s)} \chi^{0'} + \overline{(\mu_s-\gamma_s)} \chi^{1'} = -\frac m {\sqrt2} \phi_1
	\, ,
\end{aligned}
\end{equation}
where the spin coefficients are the projections of the connection coefficients onto the Newman-Penrose tetrad (see \cite[Section~4.5]{penrose_spinors_1987} or again \cite[Chapter~10]{chandrasekhar_mathematical_1998}):
\begin{align*}
\kappa_s = m^a \nabla_l l_a \,;~\rho_s = m^a \nabla_{\bar{m}} l_a 
&\,;~
\sigma_s = m^a \nabla_m l_a \,;~\tau_s = m^a \nabla_n l_a \, , 
\\
\varepsilon_s = \frac{1}{2} \left( n^a \nabla_l l_a + m^a \nabla_l \bar{m}_a \right)
&\,;~ 
\alpha_s = \frac{1}{2} \left( n^a \nabla_{\bar{m}} l_a + m^a \nabla_{\bar{m}}
\bar{m}_a \right),
\\
\beta_s = \frac{1}{2} \left( n^a \nabla_m l_a + m^a \nabla_m \bar{m}_a
\right) 
&\,;~ 
\gamma_s = \frac{1}{2} \left( n^a \nabla_n l_a + m^a \nabla_n \bar{m}_a
\right) \, , 
\\
\pi_s = - \bar{m}^a \nabla_l n_a 
\,;~ 
\lambda_s = - \bar{m}^a \nabla_{\bar{m}} n_a 
&\,;~ 
\mu_s = - \bar{m}^a \nabla_m n_a 
\,;~ 
\nu_s = - \bar{m}^a \nabla_n n_a \, .
\end{align*}

For our particular choice of tetrad \eqref{eq:tetrad}, the calculations are provided in \Cref{sec:spin coeffs}. Here we display the final result:
\begin{align}\label{eq:spin coeffs}
\nonumber
\kappa_s &= \sigma_s = \lambda_s = \nu_s = 0 \,;
&&
\tau_s = -\pi_s =
-ia\sin\theta \, \frac{r + i a \cos\theta}{\rho^3}\sqrt{\frac{\Delta_\theta}{2}} \,,	
\\[5 pt]
\rho_s &= -\mu_s 
= \frac{r+ia\cos\theta}{\rho^2} \sqrt{\frac{-\Delta_r}{2\rho^2}} \, ;
&&
\gamma_s = -\epsilon_s =
\frac{- \Delta_r'}{4\sqrt{-2 \rho^2 \Delta_r}}
+
\frac{r + ia \cos\theta}{2\rho^2} \sqrt{\frac{-\Delta_r}{2\rho^2}} \, ,
\end{align}
\begin{align*}
\alpha_s &= -\beta_s = 
-\frac{\cot\theta}{2 \rho} \sqrt{\frac {\Delta_\theta}2}
+
\frac{ia\sin\theta}{2\sqrt{2\Delta_\theta}\rho^3}\bigg(
r + ia\cos\theta 
\bigg)
\bigg(1 - \frac{i\Lambda a r\cos\theta}{3}\bigg) \, .
\nonumber
\end{align*}

To give a readable explicit expression of \Cref{eq:dirac eq NPF}, we make use of the Pauli matrices
\begin{equation}\label{eq:pauli matrices}
	\mathrm{I}_2 := 
	\begin{pmatrix}
		1 & 0 \\
		0 & 1
	\end{pmatrix};
	\quad
	\sigma_x :=
	\begin{pmatrix}
		0 & 1 \\
		1 & 0
	\end{pmatrix};
	\quad
	\sigma_y :=
	\begin{pmatrix}
		0 & -i \\
		i & 0
	\end{pmatrix};
	\quad
	\sigma_z :=
	\begin{pmatrix}
		1 & 0 \\
		0 & -1
	\end{pmatrix},
\end{equation}
and we define
\begin{align}\label{eq:Gamma matrices}
\Gamma_i := \begin{pmatrix}
\sigma_i & 0 \\
0 & -\sigma_i
\end{pmatrix},\quad
\Gamma_{ii} := \begin{pmatrix}
\sigma_i & 0 \\
0 & \sigma_i
\end{pmatrix}, \quad \text{ for } i = x, y, z;
\end{align}
in addition to the matrices
\begin{align}\label{eq:extra defs}
M_\pm := \begin{pmatrix}
		0 & \I_2 \\
		\pm \I_2 & 0	
\end{pmatrix} \quad\text{and}\quad
\Gamma_{0} := \begin{pmatrix}
\I_2 & 0 \\
0 & -\I_2	
\end{pmatrix},
\end{align}
all of which are hermitian and unitary. Finally, let
\begin{equation}\label{ftilde}
	\f(r):= \frac{-f(r)}{\sqrt{-\Delta_r(r)}} = \frac{\sqrt{-\Delta_r(r)}}{\lambda(r^2+a^2)},
\end{equation}
and note that, by \Cref{lem:asymptotics} and its remark, the asymptotic behaviour of $\f$ is given by
\begin{equation}\label{eq:f tilde asymptotics}
\f \sim e^{\kappa_\pm \tau} \qquad \text{as}~~ \tau\rightarrow\mp\infty.
\end{equation}
In this notation, \Cref{eq:dirac eq NPF} becomes
\begin{align}\label{eq:intermed dirac eq}
	\milbracket{\frac{1}{\f\sqrt{2\rho^2}} \partial_\tau 
	+ \frac{\Delta_r' \rho^2 + 2r \Delta_r}{4\sqrt{-2\Delta_r}\rho^3} 
	- \Gamma_z \frac{1}{\f\sqrt{2\rho^2}} \left(\partial_x + \frac{a}{r^2+a^2} \partial_\varphi - \frac{iqQr}{r^2+a^2}\right)& \nonumber
	\\
	+ i \Gamma_0 \frac{a\cos\theta\Delta_r}{2\sqrt{-2\Delta_r}\rho^3}
	+ \Gamma_y \frac{\lambda}{\sqrt{2\Delta_\theta\rho^2}} \left( a\sin\theta \partial_x + \frac{\partial_\varphi}{\sin\theta} \right)	
	+ i \Gamma_{xx} \frac{ra\sin\theta}{2\rho^2} \sqrt{\frac{\Delta_\theta}{2\rho^2}}
	& 
	\\
	- \Gamma_x \left(\sqrt{\frac{\Delta_\theta}{2\rho^2}}\partial_\theta + \frac{\rho^2 \cos\theta (2\Delta_\theta - \lambda) - a^2\cos\theta \sin^2\theta \Delta_\theta}{2\sqrt{2\Delta_\theta}{\sin\theta \rho^3}}\right)- M_- \frac{m}{\sqrt2}&}\Phi = 0. \nonumber
\end{align}
We note that the operator acting on $\Phi$ is regular on $(\M,\g)$. Indeed, $\Delta_r$ is strictly negative, $\Delta_\theta$ is strictly positive, and the coordinate singularity of $(\sin\theta)^{-1}$ can be amended, as usual, by transitioning to another chart on $\s^2$. However, the transition to another chart on $\s^2$ is unnecessary. Indeed, as we will see in the following section, the only appearance of $(\sin\theta)^{-1}$ will be within a regular operator (the operator $\slD$ in \eqref{eq:dirac op}).

\subsection{Hamiltonian formulation}\label{subsec:Hamil-Formulation}
Before studying the evolution of \Cref{eq:intermed dirac eq} as a Schrödinger equation, we first perform a spin transformation\footnote{See e.g. \cite{belgiorno_absence_2009} for more on such a transformation.} on $\Phi$, whose purpose is to simplify the expression of \Cref{eq:intermed dirac eq} as well as the $L^2$-norm by absorbing the density function of the volume form into $\Phi$. Set 
\begin{equation}\label{eq:spin tf}
\Psi := \A^{-1} \Phi \quad \text{ with } \quad \A = \lambda\sqrt{2} \left(-\Delta_r  \rho^2\right)^{-\frac14} S,
\end{equation}
where
\begin{equation}\label{eq:S matrix}
	S := \diag\left(e^{i\varpi}, e^{i\varpi}, e^{-i\varpi}, e^{-i\varpi}\right), \quad \varpi(r, \theta) := \frac{i}{4} \log \frac{r-ia\cos\theta}{r+ia\cos\theta}.
\end{equation}

Then, after due computation, \eqref{eq:intermed dirac eq} simplifies to
\begin{equation}\label{eq:intermed dirac eq 2}
\begin{aligned}
&\milbracket{-\A \partial_\tau + \Gamma_z \A \left(\partial_x + \frac{a}{r^2+a^2} \partial_\varphi - \frac{iqQr}{r^2+a^2} \right) - \Gamma_y \A \f \, \frac{a\sin\theta}{\sqrt{\Delta_\theta}} \left(\lambda \partial_x + \frac{\Lambda a }{3} \partial_\varphi \right)  \\
&+ \Gamma_x \A \f \frac{\Delta_\theta'}{4\sqrt{\Delta_\theta}} + \Gamma_x \A \f \sqrt{\Delta_\theta} \left(\partial_\theta + \frac{\cot\theta}2\right) - \Gamma_y \A \f \sqrt{\Delta_\theta} \frac{\partial_\varphi}{\sin\theta} + M_- \A m \f \rho } \Psi = 0.
\end{aligned}
\end{equation}
Next, we note that $S$, and hence $\A$, commutes with the matrices $\Gamma_i$ for $i = x, y, z$, and that
\begin{equation*}
M_- \A = \A \begin{pmatrix}
0 & e^{-2i\varpi} \I_2 \\
- e^{2i\varpi} \I_2 & 0
\end{pmatrix} 
= \frac{-i}{m\rho} \A M_0 ,
\end{equation*}
where
\begin{equation}	\label{eq:M op}
	M_0 := i M_- mr + a M_+ m\cos\theta.
\end{equation}
Hence, left-multiplying \eqref{eq:intermed dirac eq 2} by $\A^{-1}$ we see that $\Psi$ satisfies
\begin{equation}\label{eq:schrodinger}
\partial_\tau \Psi = -iH(\tau)\Psi \, ,
\end{equation}
where the {Dirac Hamiltonian} is given by
\begin{align}\label{eq:hamiltonian sh}
H(\tau) =
H_0(\tau) + \f(r(\tau)) \left({\sqrt{\Delta_\theta}} \slD + H_1  + M_0(\tau) \right),
\end{align}
with $M_0$ as in \eqref{eq:M op} and
\begin{gather}
\label{eq:H0 op}
H_0 := - \Gamma_z \left(D_x + \frac{a}{r^2+a^2} D_\varphi - \frac{qQr}{r^2+a^2}\right),
\\
\label{eq:Ha op}
H_1:= \Gamma_x \frac{i \Delta_\theta'}{4\sqrt{\Delta_\theta}} + \Gamma_y \frac{a\sin\theta}{\sqrt{\Delta_\theta}} \left(\lambda D_x + \frac{\Lambda a}{3} D_\varphi \right),
\\
\label{eq:dirac op}
\slD := i\Gamma_x \left(\partial_\theta + \frac{\cot\theta}{2}\right) - i \Gamma_y \frac {\partial_\varphi}{\sin\theta}.
\end{gather}
Note that the operator $\slD$ is regular on $\s^2$, including at values $\theta \in \{0, \pi\} $, when acting on $\Hsp{1}(\s^2; \C^2) \oplus \Hsp{1}(\s^2; \C^2)$ (see \eqref{eq:normequiv-DS2} in \Cref{sec:slDs}).

\begin{remark}\label{ismoetryHandHtau}
	Note that in \eqref{eq:S matrix} $\varpi \in \R$, and so $S$ is unitary. Hence, we have
	\begin{equation}\label{eq:new norm on Ht}
		\norm{\Phi}_{\mathcal{H_\tau}}^2 = \int_{\Sigma_\tau} \norm{\Psi}_{\C^4}^2 \diff x \diff \omega.
	\end{equation}
	Given our fixed tetrad $\mathbb{T}$ and its associated spin-frame, the transformation $\mathcal A$ thus, at each time $\tau$, gives an isometry between the space
	\begin{align} \label{SpaceCalH}
		\mathcal{H} = L^2 (\Sigma;\C^4 ), \quad \Vert \Psi \Vert^2_{\mathcal{H}} = \int_{\Sigma} \Vert \Psi \Vert_{\C^4}^2 \diff x \diff \omega \, ; \qquad \Sigma:=\R_x\times\s^2_\omega,
	\end{align}
	and the space $\H_\tau$ given in \eqref{SpaceHt}. Indeed, since $\Phi$ is the spin component vector of $(\phi_A,\chi^{A'})$ in the tetrad $\mathbb{T}$,  $S$ can be seen as the induced spin transformation from a (restricted) Lorentz transformation acting on the tetrad $\mathbb{T}$, and $\Psi$ is the spin component vector of the bispinor $(\phi_A,\chi^{A'})$ in the transformed tetrad. The action of $S$ is to change the spin dyad associated to the tetrad, while the dilation factor in $\A$ rescales the spinor field, yielding a density spinor.
\end{remark}

\begin{remark}\label{symmetrichamil}
	Thanks to the conservation law \eqref{eq:conserved current} for Dirac fields, \eqref{eq:new norm on Ht} means that the norm of a solution of \Cref{eq:schrodinger} is independent of $\tau$, and thus the scalar product of two such solutions is also independent\footnote{By the polarization identity and the linearity of \Cref{eq:schrodinger}.} of $\tau$. It follows that the hamiltonian $H(\tau)$ is a symmetric operator on $\H$ (whose domain we shall specify later on). Indeed, if $\Psi_1$ and $\Psi_2$ are two solutions of \Cref{eq:schrodinger}, we have
	\begin{align}
		0=\d_\tau\langle\Psi_1(\tau),\Psi_2(\tau)\rangle_\H&=\langle\d_\tau\Psi_1,\Psi_2 \rangle_\H + \langle\Psi_1,\d_\tau\Psi_2 \rangle_\H\nonumber \\
		&=-i\big(\langle H(\tau)\Psi_1,\Psi_2 \rangle_\H -\langle\Psi_1, H(\tau)\Psi_2 \rangle_\H \big). \label{eq:symmetric=unitary}
	\end{align}
\end{remark}

\subsection{The Dirac operator on $ \s^2 $}\label{sec:slDs}

One of the major difficulties for studying scattering on Kerr-type black holes is of course the lack of symmetry. The interior is no exception, and in particular, the spherical operator $\slD$ is the most delicate part of the hamiltonian in our approach and needs careful analysis. We will make use of the fact that $\slD$ can be expressed in terms of the well-known Dirac operator on the 2-sphere $\slDs$:
\begin{equation}\label{eq:dirac op on s2}
	\slD = \begin{pmatrix}
		\slDs & 0 \\
		0 & -\slDs
	\end{pmatrix}, \quad \text{where}  \quad \slDs := i \sigma_x \left(\partial_\theta + \frac{\cot\theta}2\right) - i \sigma_y \frac {\partial_\varphi}{\sin\theta}.
\end{equation}
The Dirac operator $\slDs$ has been analysed in detail in the literature and its properties are well-studied. In this section, we recall the main results and features needed for us in this work. The interested reader can refer for example to   \cite{abrikosov_jr_dirac_2002,hijazi_spectral_2001,raulot_sobolev-like_2009,camporesi_eigenfunctions_1996,bar_dirac_1996} for further discussions on $\slDs$ and its generalizations.

Let $\S(\s^2)$ be the unique spinor bundle on $\s^2$, and let $\Gamma({\S(\s^2}))$ be the space of spinor fields on $\s^2$, i.e., sections of $\S(\s^2)$. In fact, the spinor bundle $\S(\s^2)$ is trivial (see e.g. \cite{bar_dirac_1996}), and its sections can be regarded as functions from $\s^2$ into $\C^2$.  Let $L^2(\s^2; \C^2)$ be the space of square integral spinor fields on $\s^2$, whose norm we shall denote simply by $\norm{.}_{L^2(\s^2)}$:
\begin{equation*}
	\norm{\phi}_{L^2(\s^2)}^2=\int_{\s^2} \left(\abs{\phi_0}^2 + \abs{\phi_1}^2\right) \sin\theta \diff \theta \diff \varphi.
\end{equation*}

One important property of $\slDs$, is that it defines a norm which is equivalent to the first Sobolev norm of a spinor field on $\s^2$. Before stating this more explicitly, let us first recall the covariant derivatives of a spinor field  with respect to the spinorial Levi-Civita connection $\nablas$ of the negative round metric on $\s^2$, i.e., the metric $-g_{\s^2}=-\diff\omega^2= -\diff\theta^2-\sin\theta\diff\varphi^2$. The negative signature is taken to match our signature convention on the spacetime metric $\g$, $(+---)$ in which the sphere has a negative definite metric. The covariant derivatives of $\phi\in\Gamma({\S(\s^2}))$ are given by the expressions (\cite[Equation~(9)]{abrikosov_jr_dirac_2002}\footnote{Notice the "$+$" in the second equation of \eqref{eq:spin der} instead of the "$-$" appearing in \cite{abrikosov_jr_dirac_2002}. This is due to our negative signature convention.}): 
\begin{equation}\label{eq:spin der}
	\nablas_\theta \phi = \partial_\theta \phi \quad \text{and} \quad \nablas_\varphi \phi = \partial_\varphi \phi + i \sigma_z \frac{\cos\theta}2 \phi .
\end{equation}
If we instead choose the orthonormal frame $ \{e_1, e_2\} = \{ \partial_\theta, \frac{1}{\sin\theta} \partial_\varphi\} $ on $\s^2$, we have
\begin{equation}\label{eq:spin der 2}
\nablas_{e_1} \phi = \partial_\theta \phi \quad \text{and} \quad \nablas_{e_2} \phi = \frac{1}{\sin\theta} \partial_\varphi \phi+ i \sigma_z \frac{\cot\theta}2 \phi.
\end{equation}
Moreover, for such an orthonormal frame, one has, 
\begin{equation}\label{eq:lemma 4.11 hijazi}
	\Vert{\nablas \phi}\Vert_{L^2(\s^2)}^2 = \Vert{\nablas_{e_1} \phi}\Vert_{L^2(\s^2)}^2 + \Vert{\nablas_{e_2} \phi}\Vert_{L^2(\s^2)}^2.
\end{equation}

The (spinorial) Sobolev\footnote{See for example \cite{raulot_sobolev-like_2009}.} space $\Hsp{1}(\s^2; \C^2)$ is the subspace of $L^2(\s^2; \C^2)$ consisting of spinor fields $\phi$ with 
	\begin{equation*}
		\Vert{\nablas \phi}\Vert_{L^2(\s^2)} < \infty,
	\end{equation*}
	and it is endowed with the Sobolev norm $\norm{.}_{\Hsp{1}(\s^2)}$ given by
	\begin{equation*}
		\norm{\phi}_{\Hsp{1}(\s^2)}^2 = \norm{\phi}_{L^2(\s^2)}^2 + \Vert{\nablas \phi}\Vert_{L^2(\s^2)}^2.
	\end{equation*}

The higher order spinorial Sobolev spaces $\Hsp{k}(\s^2; \C^2)$, $k\in\mathbb{N}$, and the space of smooth spinors ${\mathcal C}^\infty(\s^2; \C^2)$ can be defined iteratively, in a way analogous to the usual Sobolev spaces.  Additionally, all these spaces are dense in $L^2(\s^2; \C^2)$. 

The following simple estimates are needed in later sections, especially section \ref{sec:B-operator}. We state the lemma here since it only involves the spin structure on $\s^2$.
\begin{lemma}\label{lemma:sin theta} For all $\phi \in \Hsp{1}(\s^2; \C^2)$ and for all functions $F:[0,\pi)_\theta \to \R$ such that $\abs{F(\theta)} \lesssim \abs{\sin\theta}$, we have
	\begin{align}
		&\norm{\partial_\varphi \phi}_{L^2(\s^2)} \le \norm{\phi}_{\Hsp{1}(\s^2)}, \label{eq:phi bound} \\
		&\norm{F(\theta)\left(\partial_\theta + \frac{\cot\theta}2 \right) \phi}_{L^2(\s^2)} \lesssim \norm{\phi}_{\Hsp{1}(\s^2)}. \label{eq:D-phi bound}
	\end{align}
\end{lemma}
\begin{proof}
	From \eqref{eq:spin der 2}, \eqref{eq:lemma 4.11 hijazi} and the unitarity of $\sigma_z$, we get
	\begin{align*}
		\norm{\partial_\varphi \phi}_{L^2(\s^2)}^2 &\le \norm{\left(\partial_\varphi + i \sigma_z \frac{\cos\theta}2\right) \phi}_{L^2(\s^2)}^2 + \norm{\frac{\cos\theta}2 \phi}_{L^2(\s^2)}^2 \\
		&\le \norm{\left(\frac{1}{\sin\theta}\partial_\varphi + i \sigma_z \frac{\cot\theta}2\right) \phi}_{L^2(\s^2)}^2 + \norm{\phi}_{L^2(\s^2)}^2 \le \norm{\phi}^2_{\Hsp{1}(\s^2)}.
	\end{align*}
	To show \eqref{eq:D-phi bound}, simply note that
	\begin{align*}
		\norm{F(\theta)\left(\partial_\theta + \frac{\cot\theta}2 \right) \phi}_{L^2(\s^2)}^2 & \lesssim \norm{\partial_\theta \phi}_{L^2(\s^2)}^2 + \norm{\cos\theta \, \phi}_{L^2(\s^2)}^2
		\\ & \lesssim \Vert{\nablas_{e_1}\phi}\Vert_{L^2(\s^2)}^2 + \norm{\phi}_{L^2(\s^2)}^2 \lesssim \norm{\phi}_{\Hsp{1}(\s^2)}^2.
	\end{align*}
\end{proof}

We now state the properties of $\slDs$ most relevant to us.

\begin{lemma}\label{lem:propertiesofDs2}
	The Dirac operator $\slDs$ defined in \eqref{eq:dirac op on s2} satisfies:
	\begin{enumerate}[label*=(\roman*)]
		\item $D(\slDs):=\left\{\phi\in L^2(\s^2;\C^2): \slDs\phi\in L^2(\s^2;\C^2) \right\}=\Hsp{1}(\s^2; \C^2)$.
		\item $(\slDs,D(\slDs))$ is a self-adjoint operator on ${L^2(\s^2; \C^2)}$.
		\item the Schrödinger-Lichnerowicz-Weitzenböck formula: for all $\phi\in{\mathcal C}^\infty(\s^2; \C^2)$,	
		\begin{equation}\label{eq:schrodinger-lichnerowicz}
			\slDs^2 \phi= -\nablas^2 \phi + \frac12 \phi,
		\end{equation}
	where $\nablas^2$ is the spinorial Laplacian on the sphere.
		\item For all $\phi\in\Hsp{1}(\s^2; \C^2)$, \begin{equation}\label{eq:normequiv-DS2}
			\norm{\slDs\phi}_{L^2(\s^2)} \eqsim \norm{\phi}_{\Hsp{1}(\s^2)},
		\end{equation}	
	Moreover,
	\begin{equation}\label{eq:DS2greaterthanL2}
		\norm{\slDs\phi}_{L^2(\s^2)} \ge \norm{\phi}_{L^2(\s^2)}.
	\end{equation}
	\item The spectrum of $(\slDs,D(\slDs))$ satisfies $\sigma(\slDs)\cap(-1,1)=\emptyset$. 
	\end{enumerate}
\end{lemma}
\begin{proof} For \textit{(i)}, note that $\slDs=i\sigma_x\nablas_{e_1}-i\sigma_y\nablas_{e_2}$, and hence 
	\begin{equation}\label{eq:DS2lessNablas}
		\abs{\slDs\phi}^2\le2\vert{\nablas \phi}\vert^2,
	\end{equation} 
which implies that $\Hsp{1}(\s^2; \C^2)\subseteq D(\slDs)$. The other inclusion follows from \textit{(iv)}.

The second point is a classical result and a proof in the black hole settings can be found in \cite{hafner_scattering_2004}. A treatment for manifolds more general than $\s^2$  can be found in \cite{hijazi_spectral_2001}.  

The Lichnerowicz formula is also a classical result \cite{abrikosov_jr_dirac_2002,hijazi_spectral_2001} and can be shown by a direct calculation. From \eqref{eq:dirac op on s2} we have 
\begin{equation}\label{eq:slD^2}
	\slD^2 = {D_\theta}^2 + \frac{1}{\sin^2\theta} {D_\varphi}^2 + \Gamma_{zz} \frac{\cot\theta}{\sin\theta} D_\varphi - i \cot\theta D_\theta + \frac1{4\sin^2\theta} + \frac14,
\end{equation}
and from
\begin{equation}\label{eq:laplacian-sphere-general}
	\nablas^2=\frac{1}{\sqrt{|g_{\s^2}|}}\nablas_a\left(g_{\s^2}^{ab}\sqrt{|g_{\s^2}|}\nablas_b\right)
\end{equation}
we get
\begin{equation}\label{eq:slD^2}
	-\nablas^2 = {D_\theta}^2 + \frac{1}{\sin^2\theta} {D_\varphi}^2 + \Gamma_{zz} \frac{\cot\theta}{\sin\theta} D_\varphi - i \cot\theta D_\theta + \frac1{4\sin^2\theta} - \frac14.
\end{equation}
For point \textit{(iv)}, the Formula \eqref{eq:schrodinger-lichnerowicz} gives us the equivalence of norms \eqref{eq:normequiv-DS2} for smooth spinors using \textit{(ii)} and \eqref{eq:schrodinger-lichnerowicz} with
\begin{equation}\label{eq:norm2side-inequality}
	\norm{\slDs\phi}_{L^2(\s^2)}^2 \leq \norm{\phi}_{\Hsp{1}(\s^2)}^2 \leq 2\norm{\slDs\phi}_{L^2(\s^2)}^2.
\end{equation}
Since $\Hsp{1}(\s^2; \C^2)$ is the completion of ${\mathcal C}^\infty(\s^2; \C^2)$ with respect to $\norm{.}_{\Hsp{1}(\s^2)}$, the equivalence holds for all $\phi\in\Hsp{1}(\s^2; \C^2)$ by density. Incidentally,
another simple proof for \eqref{eq:normequiv-DS2} using pseudo-differential techniques was done in \cite[Lemma~1~and~Remark~1]{raulot_sobolev-like_2009}.

One way of proving Inequality \eqref{eq:DS2greaterthanL2} is using the decomposition of $\slDs$ on spin-weighted spherical harmonics $\{{W}^l_n\}$ which form a basis for $L^2(\s^2;\C^2)$, see, e.g., \cite{gelfand_representations_1958,abrikosov_jr_dirac_2002} or \cite{hafner_scattering_2021,hafner_scattering_2004} and the references therein. In this decomposition, if\footnote{Here, $\odot$ is the component-wise Hadamard product of vectors.} 
\begin{equation}
	\phi= \sum_{(l,n)}{\phi}^l_n \odot {W}^l_n \, ,
\end{equation} 
then
\begin{equation}\label{eq:decomposion slD}
	\slDs\phi= \sum_{(l,n)}\left(l+\frac{1}{2}\right) (\sigma_x{\phi}^l_n) \odot {W}^l_n \, ,
\end{equation}
where $l\ge \frac{1}{2}$ in both sums. Taking norms, Parseval's identity gives \eqref{eq:DS2greaterthanL2} as $\sigma_x$ is unitary. Interestingly, in view of \eqref{eq:DS2lessNablas}, \eqref{eq:DS2greaterthanL2} gives a spinorial version of the usual Poincaré inequality on scalar fields. 

Finally, \eqref{eq:DS2greaterthanL2} is exactly saying that $\slDs^2\ge1$ as a quadratic form, and \textit{(v)} follows by the min-max principle. 
\end{proof}

\section{Scattering theory}\label{sec:scattering}

We now proceed to show our main result: existence, uniqueness, and asymptotic completeness of scattering states in $\H$ for  \Cref{eq:schrodinger}. To prepare for it, we state the properties of $\slD$ inherited from those of $\slDs$ in section \ref{sec:slDs}, then we introduce the simplified dynamics and the evolution systems.

Since it will be clear from the context, and to avoid cumbersome notations, we keep omitting the $\C^2$ and the $\C^4$ from the notation of the norm of a spinor ($\C^2$-valued) or a bispinor ($\C^4$-valued). For bispinors on $\S^2$, we set 
$$H^1_{\text{sp}}(\s^2; \C^4):=H^1_{\text{sp}}(\s^2; \C^2)\oplus H^1_{\text{sp}}(\s^2; \C^2),$$

and similar to \eqref{eq:spin der 2} and \eqref{eq:lemma 4.11 hijazi},  for  $\psi\in H^1_{\text{sp}}(\s^2; \C^4)$, we have
\begin{align}
	&\nablas_{e_1} \psi = \partial_\theta \psi \quad \text{and} \quad \nablas_{e_2} \psi = \frac{1}{\sin\theta} \partial_\varphi \psi+ i \Gamma_{zz} \frac{\cot\theta}2 \psi; \label{eq:spin S2 der-e-bispin} \\
	&\Vert{\nablas \psi}\Vert_{L^2(\s^2)}^2 = \Vert{\nablas_{e_1} \psi}\Vert_{L^2(\s^2)}^2 + \Vert{\nablas_{e_2} \psi}\Vert_{L^2(\s^2)}^2.\label{eq:lemma 4.11 hijazi-bispin}
\end{align}

As in the previous section, for {bispinors} on $\Sigma$, we can define $\Hsp{k}(\Sigma; \C^4)$, the scale of Sobolev dense subspaces of $\H$, as well as the dense subspace ${\mathcal C}_c^\infty(\Sigma, \C^4)$ of smooth compactly supported elements of $\H$. In particular, we have 
$$H^1_{\text{sp}}(\Sigma; \C^4) = L^2(\R_x; H^1_{\text{sp}}(\s^2; \C^4)) \cap  H^1(\R_x; L^2(\s^2; \C^4))$$ 
equipped with $\norm{.}_{H^1_{\text{sp}}(\Sigma)}$, which  for $\psi \in H^1_{\text{sp}}(\Sigma; \C^4)$, is given by
\begin{equation}\label{eq:h1 norm on sigma}
	\norm{\psi}_{H^1_{\text{sp}}(\Sigma)}^2 = \norm{\psi}_{\H}^2 + \Vert{ \nablas \psi}\Vert_{\H}^2 + \norm{\partial_x\psi}_{\H}^2.
\end{equation}
$H^1_{\text{sp}}(\Sigma; \C^4)$ can also be decomposed in a natural way as
\begin{equation}\label{eq:H1-decomp}
	H^1_{\text{sp}}(\Sigma; \C^4)=H^1_{\text{sp}}(\Sigma; \C^2)\oplus H^1_{\text{sp}}(\Sigma; \C^2).
\end{equation}
\subsection{The full dynamics} 

Our domain of definition of the hamiltonian $H(\tau)$ will be 
\begin{equation}\label{eq:D-domain}
	D:=\Hsp{1}(\Sigma; \C^4).
\end{equation} 
\begin{lemma}\label{lemma:slD estimate}
	Let $H(\tau)$ and $\slD$ be given as in \eqref{eq:hamiltonian sh}-\eqref{eq:dirac op}. Then for $\psi\in D$,
		\begin{align}
		\norm{\slD\psi}_{\H}^2 &\eqsim \norm{\psi}_{\H}^2 + \Vert{ \nablas \psi}\Vert_{\H}^2 \le \norm{\psi}_{\Hsp{1}(\Sigma)}^2, \label{eq:norm-slD}
		\\
		\norm{H(\tau)\psi}_{\H} &\lesssim \norm{\psi}_{\Hsp{1}(\Sigma)},
	\end{align}
	 and both $(H(\tau),D)$ and $(\slD,D)$ are symmetric operators on $\H$. 
\end{lemma}
\begin{proof} 
	Let $\psi = (\phi, \chi)$ in accordance with \eqref{eq:H1-decomp}. Then from \eqref{eq:normequiv-DS2}, we have that
	\begin{align*}
		\norm{\slD\psi}_{\H}^2 & = \int_{\R} \left( \norm{\slD_{\s^2}\phi}_{L^2(\s^2)}^2 + \norm{\slD_{\s^2}\chi}_{L^2(\s^2)}^2 \right) \diff x \eqsim \int_{\R} \left(\norm{\phi}_{\Hsp{1}(\s^2)}^2 + \norm{\chi}_{\Hsp{1}(\s^2)}^2 \right) \diff x  \\
		&  \le  \int_{\R}  \left(\norm{\phi}_{\Hsp{1}(\s^2)}^2 + \norm{\chi}_{\Hsp{1}(\s^2)}^2 + \norm{D_x\psi}_{L^2(\s^2)}^2 \right) \diff x = \norm{\psi}_{\Hsp{1}(\Sigma)}^2.
	\end{align*}
The second inequality follows directly form the first and from \eqref{eq:D-phi bound}, and note that the hidden constant actually does not depend on $\tau$ since $\f$ and $r$ are bounded functions of $\tau$. The fact that $H$ is symmetric was established in \Cref{symmetrichamil}, while for $\slD$ it is a direct consequence of \Cref{lem:propertiesofDs2}, \textit{(ii)}.
\end{proof}	

We now restate \Cref{lem:Well-posed-andCC} for \Cref{eq:schrodinger} and define its evolution system.

\begin{lemma} \label{HCP}
	For $s\in \R$, the initial value problem 
	\begin{equation}\label{eq:IVP}
		\begin{cases}
			&\partial_\tau \Psi (\tau) = -iH(\tau)\Psi (\tau)\, , \\
			&\Psi \vert_{\tau=s} = \psi \, , \qquad \psi \in \mathcal{H}
		\end{cases}
	\end{equation}
is well-posed, its weak solution $\Psi$ is in $\mathcal{C}(\R_\tau;\H)$, and $\Vert \Psi (\tau) \Vert_{{\mathcal H}} = \Vert \psi \Vert_{{\mathcal H}}$ for all $ \tau\in\R$. 
If $\psi \in D$, then $\Psi\in \mathcal{C}(\R_\tau;D)$ and $\d_\tau \Psi\in \mathcal{C}(\R_\tau;\H)$.

Furthermore, for $(\tau,s)\in\R^2$, the propagator
\begin{equation}
	\begin{array}{cccc}
	\mathcal{U} (\tau,s):&  \H &\longrightarrow & \H \\
	&\psi&\longmapsto&\Psi(\tau),
	\end{array}
\end{equation} 
where $\Psi$ is the solution to \eqref{eq:IVP}, is a strongly continuous unitary evolution system\footnote{See \cite{pazy_semigroups_2012} for the definition of evolution system. See also \cite[Definition~3.1]{hafner_scattering_2021}.} on $\H$. Moreover, $D$ is a stable subspace for $\mathcal{U} (\tau,s)$, and $\mathcal{U} (\tau,s)$ is a strongly continuous evolution system on $(D,\norm{.}_{\Hsp{1}(\Sigma)})$ that satisfies:
\begin{align}
	\frac{\diff}{\diff \tau} \mathcal{U}(\tau,s) \psi &= - i H(\tau) \, \mathcal{U}(\tau,s) \psi\, ,\label{eq:dtauU}\\
	\frac{\diff}{\diff s} \mathcal{U}(\tau,s) \psi &= i \, \mathcal{U}(\tau,s) H(s) \psi\, ,\label{eq:dsU}
\end{align}
for all $\psi\in D$. 
\end{lemma}
\subsection{The free dynamics}

Asymptotically as $\tau\rightarrow\pm\infty$, we will compare the dynamics of \Cref{eq:schrodinger} to the ``free'' dynamics of the equation
\begin{equation}\label{eq:freeDirac}
	\d_\tau\Psi(\tau)=-iH_0^\pm\Psi(\tau),
\end{equation} 
where the free Hamiltonians $H_0^\pm$ are
\begin{equation}\label{eq:simp hamiltonian}
	H_0^\pm := H_0\big\vert_{r = r_\mp}= - \Gamma_z
	\left(
	D_x + \frac{a}{r_\mp^2+a^2} D_\varphi - \frac{qQr_\mp}{r_\mp^2+a^2}
	\right),
\end{equation}
which is the formal limit of $H(\tau)$ as $\tau\rightarrow\pm\infty$.
\begin{remark}
	The free equation \eqref{eq:freeDirac} is a transport equation which can be explicitly solved in terms of characteristics. Whether the equation is that at $r=r_+$ or $r=r_-$ has no bearing on the form of the solution. We will thus omit the ``$\pm$" horizon indicator to ease notation. Let 
	$$
	\Omega := \frac{a}{r^2_\mp+a^2} \, , \quad \mathfrak{q} :=\frac{qQr_\mp}{r_\mp^2+a^2} \, .
	$$
	Then, given any $F_i^\pm\in\mathcal{C}^1(\R\times\s^2)$ for $i=1,2$, one can verify that
	\begin{equation}\label{eq:free soln}
	\Psi = \begin{pmatrix}
	\Psi_1^+ \\ \Psi_1^- \\ \Psi_2^- \\ \Psi_2^+
	\end{pmatrix} , \quad \Psi^\pm_i = F_i^\pm(x\pm\tau, \theta, \varphi\pm\Omega\tau) e^{\mp i\mathfrak{q}\tau} \, ,
	\end{equation}
	solves \eqref{eq:freeDirac}. Thus, given any initial data $
	F = {}^t\begin{pmatrix}
	F_1^+ & F_1^- & F_2^- & F_2^+
	\end{pmatrix}
	$ at $\tau=0$, the solution will be given by \eqref{eq:free soln}, with
	\begin{equation}
		F_i^\pm(x,\theta,\varphi) = \Psi^\pm_i(\tau=0, x, \theta, \varphi) \, .
	\end{equation}
	Therefore, the free unitary dynamics $\mathcal{U}_0 := e^{-i\tau H^\pm_0}$ (omitting, again, the ``$\pm$" horizon indicator to ease notation) can then be written as compositions of translations in $x$ and $\varphi$. Letting $T_x(\tau) G(x) := G(x+\tau)$, $T_\varphi(\tau) G(\varphi) := G(\varphi + \Omega \tau)$, we have
	\begin{equation}
		\mathcal{U}_0(\tau) = \mathrm{diag}(T(\tau), T(-\tau), T(-\tau), T(\tau)) \, , \quad T(\tau):= e^{-i\mathfrak{q}\tau}  T_x(\tau) T_\varphi(\tau) \, ,
	\end{equation}
	which is manifestly unitary on $\H$.
\end{remark}

\begin{remark}\label{rem:differentfreedynam}
	In spherically symmetric black holes, the charge term in the Dirac hamiltonian of the full dynamic has two different limits in the future and the past. In that case, this is the only reason for choosing two different comparison dynamics, one for each horizon. Here, even when the charge is zero, the rotation of the black hole also induces a non-vanishing term with different limits, manifested as $a(r_\pm+a^2)^{-1} D_\varphi$. 
	
	Interestingly, the free hamiltonians used to construct the scattering theory for Dirac fields in the exterior of a Kerr black hole in \cite[Equation~(2.57)]{hafner_scattering_2004} and in \cite[Section~IV.5.3]{daude_sur_2004}, are very similar in form to $H_0^\pm$. The difference is that there, the $\Gamma_z$ hits the radial operator only, while in our case it hits all the terms. Of course, this is because the radial and time variables exchange roles when going from the exterior to the interior of the black hole.     
\end{remark}

\begin{lemma}\label{lem:freedynamics}
	$H^\pm_0$ are self-adjoint on $\mathcal{H}$ with domain 
	$$D_0:= H^1\left(\R_x\times[0,2\pi]_\varphi;L^2((0,\pi)_\theta,\sin\theta\diff\theta;\C^4)\right)$$
	and their propagators $e^{i \tau H_0^\pm}$, with $\tau\in\R$, are strongly continuous one-parameter families of unitary operators on $\H$ which commute with $\slD^2$ 
	\begin{equation}\label{eq:slD commutes}
		\comm{\slD^2}{e^{i \tau H_0^\pm}} = 0.
	\end{equation} 
\end{lemma}
\begin{proof}
	As mentioned in \Cref{rem:differentfreedynam}, such operators were studied in detail in \cite{hafner_scattering_2004,daude_sur_2004}, and therefore we refer to the mentioned works. 
	Finally, \eqref{eq:slD commutes} is true simply since $\slD^2$ commutes with $H_0^\pm$, which can be seen from \eqref{eq:slD^2} and the fact that $\comm{\Gamma_{zz}}{\Gamma_z} = 0 $.
\end{proof}

\subsection{Main results}\label{subsec:main-result}

We are now ready to state and prove our main result. First, we have the following Proposition which is interesting by itself, but more importantly, is needed to prove the main theorem of scattering, \Cref{thm:scattering}. Section \ref{sec:B-operator} is devoted for the proof of this Proposition, which constitutes the bulk of the Theorem's proof.
\begin{proposition}\label{prop:slD u bounded}
	Let $\psi \in D=\Hsp{1}(\Sigma; \C^4)$. Then $\forall (\tau, s) \in \R^2$, we have
	\begin{equation*}
		\norm{\slD \, \mathcal{U}(\tau, s) \psi}_{\H} \lesssim \norm{\psi}_{\Hsp{1}(\Sigma)}.
	\end{equation*}
\end{proposition}

\begin{theorem} \label{thm:scattering}
	Consider the propagators $\mathcal{U}(\tau,s)$ and $e^{-i\tau H_0^\pm}$ defined as in \Cref{HCP}. Then the following strong limits exist in $\mathcal{H}$:
	\begin{eqnarray}
	W^\pm &:=& \mathrm{s-}\lim\limits_{\tau \rightarrow \pm \infty}\mathcal{U}(0,\tau)e^{-i\tau H_0^\pm},\label{Waveopslim}\\
	\Omega^\pm &:=& \mathrm{s-}\lim\limits_{\tau \rightarrow \pm \infty}e^{i \tau H_0^\pm} \mathcal{U}(\tau,0)\, .	\label{InvWaveopslim}
	\end{eqnarray}
	The operators $W^\pm$ are respectively called the future and past \underline{direct wave operators}, while the operators $\Omega^\pm$ are respectively called the future and past \underline{inverse wave operators}. These operators are unitary on $\mathcal{H}$, and inverse of each other:
	\begin{equation}
	W^\pm\Omega^\pm=\Omega^\pm W^\pm=\mathrm{Id}_{\mathcal{H}} \, .
	\end{equation}

The scattering map for charged and massive Dirac fields in the interior of a sub-extremal KN(A)dS black hole is the unitary operator on $\H$ defined by 
\begin{equation}\label{eq:scatteringop}
	S=\Omega^+W^-.
\end{equation}
\end{theorem}
\begin{proof}
	 By density, we only need to show the existence of the limits \eqref{Waveopslim} and \eqref{InvWaveopslim} for $\psi\in\mathcal{C}^\infty_c (\Sigma \, ;\C^4 )$. 
	 We will use Cook's method. For $W^\pm$, we need to show that
	\begin{equation}\label{eq:invertibleop}
		\partial_\tau\left(\mathcal{U}(0,\tau)e^{-i\tau H_0^\pm}\psi\right)\in L^1\left(\R^\pm_\tau; \H \right).
	\end{equation}
	From \Cref{eq:dsU}, we have
	\begin{equation*}
	\partial_\tau\left(\mathcal{U}(0,\tau)e^{-i\tau H_0^\pm}\psi\right)  = \mathcal{U}(0,\tau)\left( i H(\tau) - i H_0^\pm \right) e^{-i \tau H_0^\pm}\psi,
	\end{equation*}
and since $\mathcal{U}(0,\tau)$ is unitary, we have
\begin{align*}
	\left\Vert \, \mathcal{U}(0,\tau)\left(H(\tau)-H_0^\pm \right) e^{-i\tau H_0^\pm} \psi \right\Vert_{\mathcal{H}}
	=\left\Vert \left(H(\tau)-H_0^\pm \right) e^{-i\tau H_0^\pm} \psi \right\Vert_{\mathcal{H}} \, .
\end{align*}
	The difference is explicitly given by 
	\begin{align*}
	H(\tau)&-H_0^\pm =
	-\Gamma_z \left(\left(\frac{1}{r^2+a^2} - \frac{1}{r_\mp^2+a^2}\right)aD_\varphi - \left(\frac{r}{r^2+a^2} - \frac{r_\mp}{r_\mp^2+a^2} \right)qQ\right) 
	\\
	& + \f(r) \Bigg( {\sqrt{\Delta_\theta}} \slD  + M_0 + \Gamma_x \frac{i \Delta_\theta'}{4\sqrt{\Delta_\theta}} + \Gamma_y \frac{a\sin\theta}{\sqrt{\Delta_\theta}} \left(\lambda D_x + \frac{\Lambda a}{3} D_\varphi \right) \Bigg).
	\end{align*}
	%
Next, we recall that $e^{-i\tau H_0^\pm}$ is unitary, and we note that since $D_x$ and $D_\phi$ commute with $H_0^\pm$, they also commute with $e^{-i\tau H_0^\pm}$.  Using \eqref{eq:f asymptotics} and \eqref{eq:f tilde asymptotics}, and the fact that $M$ and $\Delta'_\theta$ are bounded, while $\Delta_\theta\eqsim1$, we can estimate each term in the difference:
	\begin{align*}
	& \left\Vert \left(\frac{1}{r^2+a^2} - \frac{1}{r_\mp^2+a^2}\right)a D_\varphi e^{-i\tau H_0^\pm} \psi \right\Vert_\mathcal{H} \lesssim \abs{r-r_\mp} \Big\Vert e^{-i\tau H_0^\pm} D_\varphi  \psi \Big\Vert_\mathcal{H} \sim e^{2\kappa_\mp \tau} \Big\Vert \partial_\varphi \psi \Big\Vert_\mathcal{H} \, ,
	\\
	& \left\Vert iqQ \left(\frac{r}{r^2+a^2} - \frac{r_\mp}{r_\mp^2+a^2}\right) e^{-i\tau H_0^\pm} \psi \right\Vert_\mathcal{H} \lesssim   \abs{r-r_\mp} \Big\Vert e^{-i\tau H_0^\pm} \psi \Big\Vert_\mathcal{H} \sim e^{2\kappa_\mp \tau} \Big\Vert \psi \Big\Vert_\mathcal{H} \, ,
	\\
	& \left\Vert \f M_0 e^{-i\tau H_0^\pm} \psi \right\Vert_\mathcal{H} \lesssim \f \, \Big\Vert e^{-i\tau H_0^\pm} \psi \Big\Vert_\mathcal{H} \sim e^{\kappa_\mp \tau} \Big\Vert \psi \Big\Vert_\mathcal{H} \, ,
	\\
	&
	\left\Vert \f\,\Gamma_x \frac{i \Delta_\theta'}{4\sqrt{\Delta_\theta}} e^{-i\tau H_0^\pm} \psi\right\Vert_\mathcal{H} \lesssim \f \, \Big\Vert e^{-i\tau H_0^\pm} \psi \Big\Vert_\mathcal{H} \sim e^{\kappa_\mp \tau} \Big\Vert \psi \Big\Vert_\mathcal{H} \, , 
	\\
	& \left\Vert \f \, \Gamma_y \frac{a\lambda\sin\theta}{\sqrt{\Delta_\theta}}  D_x e^{-i\tau H_0^\pm} \psi \right\Vert_\mathcal{H} \lesssim \f \, \Big\Vert e^{-i\tau H_0^\pm} D_x \psi \Big\Vert_\mathcal{H} \sim e^{\kappa_\mp \tau} \Big\Vert D_x  \psi \Big\Vert_\mathcal{H} \, ,
	\\
	& \left\Vert \f \, \Gamma_y \frac{\Lambda a^2\sin\theta}{3\sqrt{\Delta_\theta}} D_\varphi e^{-i\tau H_0^\pm} \psi \right\Vert_\mathcal{H} \lesssim \f \, \Big\Vert e^{-i\tau H_0^\pm} D_\varphi \psi \Big\Vert_\mathcal{H} \sim e^{\kappa_\mp \tau} \Big\Vert D_\varphi  \psi \Big\Vert_\mathcal{H} \, ,
	\end{align*}
	as $\tau \to \pm \infty$, and using \eqref{eq:slD commutes}
	\begin{align*}
	\left\Vert \f \, {\sqrt{\Delta_\theta}} \slD e^{-i\tau H_0^\pm} \psi  \right\Vert_\mathcal{H} &\lesssim \f \, \abs{{\inner{\slD e^{-i\tau H_0^\pm} \psi}{\slD e^{-i\tau H_0^\pm} \psi}}_\H}^{\frac12} \\ &=\f \abs{{\inner{\slD^2 e^{-i\tau H_0^\pm}\psi}{e^{-i\tau H_0^\pm}\psi}}_\H}^{\frac12} \\
	&=\f \norm{\slD\psi}_\mathcal{H} 
	\sim e^{\kappa_\mp \tau} \norm{\slD\psi}_\mathcal{H} \, .
	\end{align*}
	Since $\psi  \in \mathcal{C}^\infty_c (\Sigma \, ; \C^4)$, all the norms on the right hand sides are finite. Therefore, all of the above terms decay exponentially to zero as $\tau \rightarrow \pm\infty$, giving integrability as required. 
	 
	For $\Omega^\pm$, similarly, we need to establish that
	\[ \partial_\tau\left(e^{i\tau H_0^\pm}\mathcal{U}(\tau, 0)\psi\right)\in L^1\left(\R^\pm_\tau; \H \right) \, .\]
	As previously, we have
	\begin{equation*}
	\partial_\tau\left(e^{i\tau H_0^\pm}\mathcal{U}(\tau, 0)\psi\right)  = -e^{i \tau H_0^\pm}\left( i H(\tau) - i H_0^\pm \right)\mathcal{U}(\tau, 0) \psi.
	\end{equation*}
	Note that every term but the $\slD$ term of the difference of the Hamiltonians is estimated identically as for $W^\pm$. This is the purpose of \Cref{prop:slD u bounded}, from which we get
	\begin{align}\label{eq:DU-decay-estimate}
	\left\Vert e^{i\tau H_0^\pm} \f \, {\sqrt{\Delta_\theta}} \slD \, \mathcal{U}(\tau, 0) \psi  \right\Vert_\mathcal{H} \lesssim \f \, \Big\Vert \slD \, \mathcal{U}(\tau, 0) \psi \Big\Vert_\mathcal{H} \lesssim \f \Big\Vert  \psi \Big\Vert_{\Hsp{1}(\Sigma)} \sim e^{\kappa_\mp \tau} \Big\Vert  \psi \Big\Vert_{\Hsp{1}(\Sigma)} \, .
	\end{align}

We have shown that the operators $W^\pm$ and $\Omega^\pm$ are well-defined as \eqref{Waveopslim} and \eqref{InvWaveopslim} on $\mathcal{C}^\infty_c (\Sigma \, ;\C^4 )$, and clearly they are inverse of each other. Since they are the strong limits of unitary operators they are partial isometries, and hence they extend to all of $\H$, still defined by \eqref{Waveopslim} and \eqref{InvWaveopslim}. Since they are invertible, i.e., \eqref{eq:invertibleop} holds, they are unitary on $\H$.   
\end{proof}

\begin{corollary}\label{cor:regularity transfer}
	Since $D_x$ and $D_\varphi$ commute with both Hamiltonians, and hence with their dynamics, they also commute with the wave operators and their inverses as well as the scattering map. In the geometric picture, this means that regularity in these directions is transferred from the horizon to the horizon. 
\end{corollary}

\begin{remark}
	However, for the other angular derivative and the transverse direction, more analysis is required and we postpone it to a future work together with the geometric interpretation. 
\end{remark}

\begin{remark}	
	Given a gauge transformation as in \eqref{eq:gauge on A}, the scattering operator transforms as
	\begin{equation}
	S \mapsto \tilde{S} = e^{iq\xi} S e^{-iq\xi} \, .
	\end{equation}
	Such a transformation will be important when dealing with the geometric interpretation (the characteristic Cauchy problem approach) since our potential $A$ is not smooth on the horizons. However, as mentioned earlier, we postpone this discussion for a future work.
\end{remark}

\section{Symmetry operator}
In this section we provide an alternative and simpler proof of \Cref{thm:scattering} using a symmetry operator for the Dirac fields in spacetimes under consideration\footnote{It is worth mentioning that the proof using some symmetry operator was suggested by one of the referees and only afterwards we found the exact symmetry operator.}. The simplification occurs in the proof via Cook's method and bypasses the use of \Cref{prop:slD u bounded}. Although the proof using the symmetry operator is simpler, we do not expect it to be robust under perturbations of the fields, while we expect the proof using \Cref{prop:slD u bounded} to be robust and possibly more apt for studying perturbations. 

Here, we only use this symmetry operator to prove scattering, however we expect it to have important applications in other contexts. To our knowledge, this is the first time the symmetry operator we introduce here appears in the literature, and may have been unknown before.

To define our symmetry operator we first denote
\begin{align}
\slcalD :=& \sqrt{\Delta_\theta}\slD + H_1  \nonumber \\
=& i\sqrt{\Delta_\theta} \Gamma_x
\left({\partial_\theta} + \frac{\cot\theta}{2} +
\frac{\Delta_\theta'}{4{\Delta_\theta}} \right) - \frac{i \lambda}{\sqrt{\Delta_\theta} \sin\theta}  \Gamma_{y} \left(a\sin^2\theta \d_x + \d_\varphi\right)
\label{eq:slcalD}
\\
=:&
\Gamma_x \T_x + \Gamma_y \T_y
\, ,	\label{eq:tt} 
\end{align}
and note that it can also be expressed as
\begin{equation}
\slcalD = i\sqrt{\Delta_\theta} \Gamma_x
	\left({\partial_\theta} + \frac{\cot\theta}{2} +
	\frac{\Delta_\theta'}{4{\Delta_\theta}} - \frac{i \lambda \Gamma_{zz}}{\Delta_\theta \sin\theta} \left(a\sin^2\theta \d_x + \d_\varphi\right)
	\right) \, .		
\end{equation}
Now we define the operator
\begin{equation}\label{eq:Q}
\Q := \left(\slcalD + M_+ma\cos\theta \right)^2 \, .
\end{equation}
\begin{theorem}[Commuting symmetry operator]\label{thm:symm op}
	The operator $\Q$ commutes with the Hamiltonian $H$ defined in \eqref{eq:hamiltonian sh}, and hence with the unitary dynamics $\mathcal{U}(\tau,s)$ for all $(\tau, s) \in \R^2$. It follows that $\Q$ is a symmetry operator, i.e. it sends solutions of the Dirac equations to other solutions. Moreover, for all $\psi \in D$, we have
		\begin{equation}\label{eq: norm-identity}
			\norm{(\slcalD+M_+ma\cos\theta)\psi}^2_\H =  \inner{\Q\psi}{\psi} \, .
		\end{equation}  
\end{theorem}

\begin{proof}
First, we show that $ \slcalD $ is a symmetric operator on $D$ the domain of $H$ with respect to the $\H$-norm. Let $\alpha, \beta \in D$. We note that
\begin{align*}
	\inner{\alpha}{\slcalD\beta}_\H &= \inner{\alpha}{\T_x \Gamma_x\beta}_\H +  \inner{\alpha}{\T_y \Gamma_y\beta}_\H \, .
\end{align*}
by unitarity of $\Gamma_i$. It is enough then to show that $\T_x$ and $\T_y$ are symmetric. Since $i\d_x$ and $i\d_\varphi$ are symmetric on $D$, it immediately follows that $\T_y$ is symmetric. It remains to show that $\T_x$ is symmetric on $D$. Indeed, we have, say for $\gamma = \Gamma_x\beta$,
\begin{align*}
	\inner{\alpha}{\T_x \gamma}_\H =& i \iint_{\R\times [0,2\pi]} dx d\varphi \int_{0}^\pi  \bar\alpha
	\left({\partial_\theta} + \frac{\cot\theta}{2} +
	\frac{\Delta_\theta'}{4{\Delta_\theta}} \right) \gamma \, \sqrt{\Delta_\theta} \sin\theta d\theta \\
	=& -i \iint_{\R\times [0,2\pi]} dx d\varphi \int_{0}^\pi \left({\partial_\theta} - \frac{\cot\theta}{2} -
	\frac{\Delta_\theta'}{4{\Delta_\theta}} \right) \bar\alpha
	\gamma \, \sqrt{\Delta_\theta} \sin\theta d\theta \\
	& -i \iint_{\R\times [0,2\pi]} dx d\varphi \int_{0}^\pi \bar\alpha
	\gamma \, \d_\theta \left[ \sqrt{\Delta_\theta} \sin\theta d\theta \right] = \inner{\T_x \alpha}{\gamma}_\H \, .
\end{align*}

Next, let us write the Hamiltonian as
\begin{equation}
	H = H_0 + imr\f M_- + \f(\slcalD+M_+ma\cos\theta) \, . \label{eq:rexpr}
\end{equation}
Evidently, $\Q$ commutes with the last term. It remains to verify that the matrices appearing in $\Q$ commute with $\Gamma_z$ and $M_-$, since the first two terms of \eqref{eq:rexpr} are $(x,\theta,\varphi)$-independent, and $\Q$ has no $\tau$ derivatives. To this end, we express
\begin{equation}\label{eq:anticomm}
	\Q = \slcalD^2 + m^2a^2\cos^2\theta + ma\{\slcalD, M_+\cos\theta\} \, ,
\end{equation}
where $\{,\}$ denotes the anti-commutator. We now make explicit the matrices appearing in \eqref{eq:anticomm}. Using \eqref{eq:tt}, we see that
\begin{align*}
	\slcalD^2 &= \T_x^2 + \T_y^2 + \{\T_x\Gamma_x, \T_y\Gamma_y\} 
	\\
	&= \T_x^2 + \T_y^2 + \Gamma_x\Gamma_y \T_x\T_y + \Gamma_y \Gamma_x \T_y\T_x
	\\
	&= \T_x^2 + \T_y^2 + i\Gamma_{zz} [\T_x,\T_y] \, ,
\end{align*}
due to $\Gamma_x\Gamma_y = - \Gamma_y\Gamma_x = i\Gamma_{zz}$. We will denote
\begin{equation*}
	\T_z := i[\T_x,\T_y] = -\sqrt{\Delta_\theta} \left(\d_\theta\T_y \right) \, ,
\end{equation*}
whose explicit expression is not pertinent. Hence,
\begin{equation}\label{eq:slcalD^2}
	\slcalD^2 = \T_x^2 + \T_y^2 + \Gamma_{zz} \T_z
\end{equation}
is diagonal. As for the final term, we note that
\begin{align*}
	\{\slcalD, M_+\cos\theta\} = \Gamma_xM_+\T_x\cos\theta + M_+\Gamma_x\cos\theta\T_x = \Gamma_xM_+ [\T_x,\cos\theta] \, ,
\end{align*}
since $\Gamma_xM_+ = -M_+\Gamma_x$, and thus

\begin{equation*}
	\{\slcalD, M_+\cos\theta\} = \Gamma_xM_+ i\sqrt{\Delta_\theta}\cos\theta \, .
\end{equation*}
This allows us to write
\begin{equation}\label{eq:explicit Q}
	\Q = \T_x^2 + \T_y^2 + m^2a^2\cos^2\theta + \Gamma_{zz} \T_z + \Gamma_xM_+ ima\sqrt{\Delta_\theta}\cos\theta \, .
\end{equation}
It is now a trivial computation to show that
\begin{gather*}
	[\Gamma_{zz},\Gamma_z] = 0 \, , \quad  [\Gamma_{zz},M_-] = 0 \, , \\
	[\Gamma_xM_+,\Gamma_z] = 0 \, , \quad  [\Gamma_xM_+,M_-] = 0 \, ,
\end{gather*}
whence $[\Q,H] = 0$.
\end{proof}

We are now poised to give a simpler proof of \Cref{thm:scattering}.
\begin{proof}[Proof of \Cref{thm:scattering} via $\Q$]
Using the conserved current, we have for all $\tau\in\R$
\begin{equation}\label{eq:symmconv}
\norm{\mathcal{Q}\mathcal{U}(\tau,0)}_\H = \norm{\mathcal{Q}\psi}_\H \, .
\end{equation}
Hence, in Cook's method used for the proof of the existence of the strong limit corresponding to $\Omega^{\pm}$, we have
\begin{align*}
\norm{\partial_\tau\left(e^{i\tau H_0^\pm}\mathcal{U}(\tau, 0)\psi\right)}^2_\H  \le & \norm{\left(H_0 - H_0^\pm \right)\mathcal{U}(\tau, 0)\psi}_\H^2 +
\f^2 \norm{M_0\mathcal{U}(\tau, 0)\psi}_\H^2 \\ &+ \f^2 \norm{\left(\sqrt{\Delta_\theta}\slD + H_1 + M_+ ma\cos\theta \right)\mathcal{U}(\tau, 0)\psi}_\H^2 \, .
\end{align*}
The first two terms are controlled as before, whereas the last term, as per \eqref{eq: norm-identity}, reads
\begin{equation*}
\Big\Vert{\Big( \underbrace{\sqrt{\Delta_\theta}\slD + H_1}_{\slcalD} + M_+ ma\cos\theta\Big)\mathcal{U}(\tau, 0)\psi}\Big\Vert_\H^2 = \inner{\mathcal{Q}\mathcal{U}(\tau, 0)\psi}{\mathcal{U}(\tau, 0)\psi}_\H = \inner{\mathcal{Q}\psi}{\psi}_\H \, .
\end{equation*}
using the polarization identity \eqref{eq:symmetric=unitary} and then \eqref{eq:symmconv}.
\end{proof}

\begin{remark}
	Note that in the case of Reissner-Nordström-like spacetimes previous considered in \cite{hafner_scattering_2021-1}, the symmetry operator $\mathcal{Q}$ reduces to the Dirac operator on the sphere $\slD$. Therefore,  \Cref{thm:symm op} provides a shortcut in the part using Cook's method, and hence a simpler proof for the scattering theory constructed therein.  
\end{remark}

\section{Commutator estimates}\label{sec:B-operator}

The aim of this section is to prove \Cref{prop:slD u bounded}. This is essentially an application of Grönwall's lemma, but the core of the proof relies on careful commutator estimates. This uses techniques adapted for time-dependent hamiltonians from the theory of regularly generated dynamics, see  \cite{derezinski_scattering_1997}.

\subsection{The comparison operator}
The idea is to introduce a comparison operator $B$, which dominates the $\slD$ operator and determines the norm on $D=\Hsp{1}(\Sigma;\C^4)$ in a way that allows us to better control the spherical derivatives of the dynamics $\mathcal{U}(\tau,s)$.

Let $(B, D(B))$  be given by
\begin{equation}\label{eq:comparison operator B}
D(B) := \Hsp{2}(\Sigma; \C^4), \quad B := {D_x}^2 + \slD^2,
\end{equation}
Seemingly simple, the (diagonal) operator $B$ had to be chosen carefully, not only to possess the above mentioned criteria, but also to suit the asymptotics of the Hamiltonian $H(\tau)$ as $\tau\rightarrow\pm\infty$. See \Cref{eq:H_0B=0}.
 
It is clear from the discussions in the previous sections that $B$ is a positive self-adjoint operator on $\H$ with a bounded inverse since its spectrum does not contain 0. Moreover, $D(B^{1/2})=D$ and $B^{-1/2}$ is a bounded operator on $\H$. Indeed, for $\psi\in D$, from \eqref{eq:norm-slD} and \eqref{eq:h1 norm on sigma}, we have
\begin{equation}\label{eq:norm equiv Bm}
	\norm{B^{1/2}\psi}^2_{\H} = {\inner{B\psi}{\psi}}_\H = {\inner{{D_x}\psi}{D_x\psi}}_\H + {\inner{\slD\psi}{\slD\psi}}_\H\eqsim \norm{\psi}_{\Hsp{1}(\Sigma)}^2,
\end{equation}
and since $\Bm$ is a bijection from $\H$ onto $D$, we have shown that for all $\psi\in\H$
	\begin{equation}\label{eq:b-1/2 equiv to H}
	\norm{\Bm\psi}_{\H} \le \norm{\Bm\psi}_{\Hsp{1}(\Sigma)} \eqsim \norm{\psi}_{\H}.
\end{equation} 

The next estimates will be used repeatedly and it is useful to collect them in a lemma.
\begin{lemma}\label{lemma:B norm}
	The following inequality relations hold
		\begin{equation}\label{eq:bounded operators Bm}
			 \norm{\slD B^{-1/2}}_{\B(\H)} \le 1 ;\quad
			 \norm{D_x B^{-1/2}}_{\B(\H)} \le 1 ;\quad
			 \norm{D_\varphi B^{-1/2}}_{\B(\H)} \le 1 .
		\end{equation}
	Furthermore, for all $\psi\in D$ we have, 
	\begin{equation}\label{eq:Bm bounded operators}
			\norm{B^{-1/2} \slD \psi}_{\H} \le \norm{\psi }_\H ;~~
			\norm{B^{-1/2} D_x \psi}_{\H} \le \norm{\psi }_\H ; ~~
			\norm{B^{-1/2} D_\varphi \psi}_{\H} \le \norm{\psi }_\H.
		\end{equation}
	and for any $\mathcal{C}^1$ function $F:[0,\pi)_\theta \rightarrow \R$, such that $\abs{F(\theta)}\lesssim \abs{\sin\theta}$ and $\abs{F'(\theta)}\lesssim 1$,
	\begin{equation}\label{eq:Bm bounded d-theta}
			\norm{ F(\theta)\left(\d_\theta + \frac{\cot \theta}{2}\right)\Bm}_{\mathcal{B}(\H)}\lesssim 1
		~~ ; ~~
			\norm{\Bm F(\theta)\left(\d_\theta + \frac{\cot \theta}{2}\right)\psi}_\H\lesssim \norm{\psi}_\H .
	\end{equation}
\end{lemma}
\begin{proof}
	\eqref{eq:bounded operators Bm} follows immediately from \eqref{eq:norm equiv Bm}, \eqref{eq:phi bound} and the fact that $\Bm$ is a bijection from $\H$ to $D$. 
As $\slD$, $D_x$ and $D_\varphi$ commute with $B$, \eqref{eq:Bm bounded operators} follows from \eqref{eq:bounded operators Bm}. To prove the two estimates in  \eqref{eq:Bm bounded d-theta}, we use \eqref{eq:bounded operators Bm} and \eqref{eq:Bm bounded operators} with the following two identities,
\begin{align*}
	F(\theta)\left(\d_\theta + \frac{\cot \theta}{2}\right)&=-iF(\theta)\Gamma_x\slD -\Gamma_{zz}\frac{F(\theta)}{\sin\theta}D_\varphi
	\\
	&=-i\slD F(\theta)\Gamma_x-F'(\theta)+D_\varphi\Gamma_{zz}\frac{F(\theta)}{\sin\theta},
\end{align*}
respectively.
\end{proof}

The following property of $B$ is simple but necessary to justify some steps in the calculations afterwards.

\begin{lemma}\label{lem:property of B}
Let $T_1\le T_2$ be two real numbers, then $\forall \tau\in\R$, the map
$s \mapsto \mathcal{U}(\tau,s) B^{-1} $ is an element of $W^{1,1}([T_1, T_2]; \B(\H))$.
\begin{proof}
	 Since $\mathcal{U}(\tau, s)$ is unitary on $\H$, we have $\forall \tau \in \R,$
	\begin{equation*}
	 \int_{T_1}^{T_2} \norm{\,\mathcal{U}(\tau, s) B^{-1}}_{\B(\H)} \diff s = \left({T_2}-{T_1}\right) \norm{B^{-1}}_{\B(\H)}.
	\end{equation*}
	Moreover, since $B^{-1}(\H)=\Hsp{2}(\Sigma;\C^4)\subset D$ we can use \Cref{eq:dsU}, and thanks to the unitarity of the $\Gamma$ matrices, we have
	\begin{align*}
		\norm{\partial_s\,\mathcal{U}(\tau, s) B^{-1}}_{\B(\H)} &= \norm{\,\mathcal{U}(\tau, s) i H(s) B^{-1}}_{\B(\H)} = \norm{H(s) B^{-1}}_{\B(\H)} \\ &= \norm{\left(H_0(s) + \f(r(s))\left(\sqrt{\Delta_\theta}\slD + H_1 + M_0(s)\right)\right)B^{-1/2}B^{-1/2}}_{\B(\H)} \\ &\lesssim \norm{H_0(s)B^{-1}}_{\B(\H)} + \f(r(s)) \norm{ \left(\sqrt{\Delta_\theta}\slD + H_1 + M_0(s)\right)B^{-1}}_{\B(\H)} \\ & \lesssim \norm{B^{-1/2}}_{\B(\H)}\left( 1 + \f(r(s))\right)  \lesssim 1+\f(r(s)),
	\end{align*}
	where \Cref{lemma:B norm} was used, in conjunction with $|\Delta_\theta'| \lesssim |\sin\theta| $ for $H_1$, and finally, Equation \eqref{eq:b-1/2 equiv to H} gives the upper bound. Therefore,
	\begin{equation*}
	 \int_{T_1}^{T_2} \norm{\partial_s\,\mathcal{U}(\tau, s) B^{-1}}_{\B(\H)} \diff s \lesssim \int_{T_1}^{T_2} \left( 1 + \f(r(s))\right)\diff s\le T_2 - T_1 + \int_{\R} \f(r(s))\diff s,
	\end{equation*}	
which is finite by the asymptotics \eqref{eq:f tilde asymptotics}. 
\end{proof}
\end{lemma}
\subsection{Proof of \Cref{prop:slD u bounded}}
The most technical part of the proof of \Cref{prop:slD u bounded} is a commutator estimate between $H(\tau)$ and $B$. For clarity, we first provide separately some commutation relations that are used frequently in the main calculation.
   
\begin{lemma}\label{lemma:comms} Let $ F \in \mathcal{C}^2([0,\pi)_\theta;\R)$. We have the following identities
	\begin{align*}
	&\comm{\slD}{F(\theta)} = i\Gamma_x F'(\theta); \quad \comm{\slD^2}{F(\theta)} = -2F'(\theta)\left(\partial_\theta + \frac{\cot\theta}2 \right) - F''(\theta);
	\\
	&\comm{\slD^2}{\Gamma_x} = 2i \, \Gamma_y \frac{\cot\theta}{\sin\theta} D_\varphi; \quad \comm{\slD^2}{\Gamma_y} = -2i \, \Gamma_x \frac{\cot\theta}{\sin\theta} D_\varphi; \quad 
	\comm{\slD^2}{\Gamma_z} = 0;  \quad \comm{\slD^2}{M_\pm} = 0.
	\end{align*}
\end{lemma}
\begin{proof}
	These are straightforward calculations which follow from the expressions of $\slD$ in  \eqref{eq:dirac op} and $\slD^2$ in \eqref{eq:slD^2} as well as the identities
	\begin{align*}
	\comm{\Gamma_{zz}}{\Gamma_x} = 2i \Gamma_y, && \comm{\Gamma_{zz}}{\Gamma_y} = - 2i \Gamma_x, && \comm{\Gamma_{zz}}{\Gamma_z} = 0, &&& \comm{\Gamma_{zz}}{M_\pm} = 0.
	\end{align*}
\end{proof}

\begin{lemma}\label{lem:H-B com}
	The map $\tau\mapsto B^{-1/2}\comm{H(\tau)}{B}B^{-1/2}$ on $\R_\tau$, where for each $\tau$ the image is first defined as a quadratic form on $D=D(B^{1/2})$, extends to an operator-valued map that belongs to $L^1(\R_\tau; B(\H))$.
\end{lemma}
\begin{proof}
	 As $D(B^{1/2}) = D(H(\tau))= D $, we have for $\psi \in D$,
	\begin{align*}
	\inner{\Bm \comm{H(\tau)}{B}\Bm \psi}{\psi}_\H &= \inner{\comm{H(\tau)}{B}\Bm \psi }{\Bm \psi}_\H \\&= \inner{B^{1/2}\psi}{H(\tau)\Bm \psi}_\H - \inner{H(\tau)\Bm \psi}{B^{1/2}\psi}_\H
	\end{align*}
	which means the quadratic form is well-defined. 
	
	We start by calculating the commutator for each term of $H(\tau)$ in \eqref{eq:hamiltonian sh} using \Cref{lemma:comms} when required. Clearly,
	\begin{equation*}
		\comm{H}{B}=\comm{H}{\slD^2}.
	\end{equation*} Particularly, for $H_0$, we have
	\begin{align*}
	\comm{H_0}{\slD^2} = \left(D_x + \frac{a}{r^2+a^2} D_\varphi - \frac{qQr}{r^2+a^2}\right) \comm{-\Gamma_z}{\slD^2} = 0.
	\end{align*}
\begin{remark}\label{eq:H_0B=0}
The fact that $\Bm[H_0,B]\Bm$ vanishes identically and is not merely bounded is essential for our method to work. As $H_0$ is the only part of the Hamiltonian that is not multiplied by the exponentially decaying function $\f$, it would otherwise give at best an exponential bound instead of the constant appearing in \eqref{eq:Gronwallconstant} after using Grönwall's lemma, and consequently, competing with the exponential decay in the Estimate \eqref{eq:DU-decay-estimate}. Of course, this would have also rendered the statement of the current lemma incorrect.  
\end{remark}
Since $\f$ commutes with $\slD^2$, we have
\begin{equation*}
	\comm{H}{\slD^2}=\f \left(\comm{\sqrt{\Delta_\theta} \slD}{\slD^2}+\comm{M_0}{\slD^2} +\comm{H_1}{\slD^2}\right)=:\f(C_1+C_2+C_3).
\end{equation*} 
We therefore calculate:
	\begin{align*}
	C_1:=\comm{\sqrt{\Delta_\theta} \slD }{\slD^2} &=
	\comm{\sqrt{\Delta_\theta}}{\slD^2} \slD =
	\left(\frac{\Delta_\theta'}{\sqrt{\Delta_\theta}}\left(\partial_\theta + \frac{\cot\theta}2 \right) + \left(\sqrt{\Delta_\theta}\right)''\right) \slD,
	\end{align*}
as well as
	\begin{align*}
	C_2:=\comm{M_0}{\slD^2} &= ma M_+ \comm{\cos\theta}{\slD^2} = -ma M_+ \left(2\sin\theta\left(\partial_\theta + \frac{\cot\theta}2 \right) + \cos\theta \right).
	\end{align*}
To obtain $\comm{H_1}{\slD^2}$, we split the commutator into two parts. Set
	\begin{align*}
		C_3:=\comm{H_1}{\slD^2} &= \comm{\Gamma_x \frac{i \Delta_\theta'}{4\sqrt{\Delta_\theta}}}{\slD^2} + \left(a\lambda D_x + \frac{\Lambda a^2}{3} D_\varphi\right)\bigg[\frac{a\sin\theta}{\sqrt{\Delta_\theta}}  \Gamma_y ,{\slD^2}\bigg]\\
		&=:C_{3a}+\left(a\lambda D_x + \frac{\Lambda a^2}{3} D_\varphi\right)C_{3b}.
	\end{align*}
	The first part can be written as 
	\begin{equation}\label{eq:gamma x term}
		C_{3a}=
		\slD\tilde{C}_{3a} + \tilde{C}_{3a}\slD,
	\end{equation}
where
	\begin{align*}
	\tilde{C}_{3a}:=\comm{\Gamma_x \frac{i \Delta_\theta'}{4\sqrt{\Delta_\theta}}}{\slD} = \frac14 \left( \left(\frac{{\Delta_\theta}'}{\sqrt{\Delta_\theta}}\right)' + 2i \Gamma_{zz} \frac{{\Delta_\theta}'}{\sqrt{\Delta_\theta}} \frac{\partial_\varphi}{\sin\theta} \right).
	\end{align*}
For the second part, we have 
\begin{align*}
	C_{3b}&= \Bigg[\hspace{-4 pt}\Bigg[2 \, \Gamma_x \frac{\cot\theta}{\sqrt{\Delta_\theta}} \partial_\varphi + 2\Gamma_y \frac{\cos\theta}{\sqrt{\Delta_\theta}}\left(\partial_\theta + \frac{\cot\theta}2 \right) \\
	&\hspace{3.27cm} + 2\Gamma_y \left(\frac{1}{\sqrt{\Delta_\theta}}\right)' \sin\theta \left(\partial_\theta + \frac{\cot\theta}2 \right)  + \Gamma_y \left(\frac{\sin\theta}{\sqrt{\Delta_\theta}}\right)'' \Bigg]\hspace{-4 pt}\Bigg]
	\\
	&=\Bigg[\hspace{-4 pt}\Bigg[ \frac{-2i\cos\theta}{\sqrt{\Delta_\theta}}\Gamma_{zz} \left(-\Gamma_y \frac{1}{\sin\theta} \partial_\varphi + \Gamma_x \left(\partial_\theta + \frac{\cot\theta}2 \right) \right)\\
	&\hspace{3.12cm}+ \Gamma_y \left(2 \sin\theta \left(\frac{1}{\sqrt{\Delta_\theta}}\right)' \left(\partial_\theta + \frac{\cot\theta}2 \right)  + \left(\frac{\sin\theta}{\sqrt{\Delta_\theta}}\right)'' \right) \Bigg]\hspace{-4 pt}\Bigg]
	\\
	&=\Bigg[\hspace{-4 pt}\Bigg[ \frac{-2\cos\theta}{\sqrt{\Delta_\theta}}\Gamma_{zz} \slD + \Gamma_y \left(2 \sin\theta \left(\frac{1}{\sqrt{\Delta_\theta}}\right)' \left(\partial_\theta + \frac{\cot\theta}2 \right)  + \left(\frac{\sin\theta}{\sqrt{\Delta_\theta}}\right)'' \right) \Bigg]\hspace{-4 pt}\Bigg].
	\end{align*}
	
	We now estimate each of the above commutator terms sandwiched by $\Bm$ to obtain the required bound on 
	$$\norm{B^{-1/2}\comm{H(\tau)}{B}B^{-1/2}}_{\mathcal{B}(\H)}.$$ 
	Recall that $\Delta_\theta \eqsim 1$ and note that $\abs{\Delta_\theta''} \lesssim 1$, in addition to $\Delta_\theta'$ being admissible as $F$ in \Cref{lemma:sin theta} since $\abs{\Delta_\theta'} \lesssim \abs{\sin\theta}$. Finally, we will  repeatedly use \Cref{lemma:B norm} and the unitarity of the $\Gamma$ matrices. Let $\psi\in\H$, from the above calculations, we have:
	
\noindent First,
	\begin{align*}
	\Bnorm{\f C_1}{\psi}_\H&\lesssim\f \, \Bnorm{\sin\theta \left(\partial_\theta + \frac{\cot\theta}2 \right) \slD}{\psi}_{\H}\\
	&\hspace{3cm}+ \f \, \Bnorm{\left(\sqrt{\Delta_\theta}\right)''\slD}{\psi}_{\H} \\
	&\lesssim \f \norm{\psi}_{\H}.
	\end{align*}
Then, 
	\begin{align*}
	& \Bnorm{\f C_2}{\psi}_\H\lesssim \f \norm{ \left(2\sin\theta\left(\partial_\theta + \frac{\cot\theta}2 \right) + \cos\theta \right)\Bm\psi}_{\H} \lesssim \f \norm{\psi}_{\H}.
	\end{align*}
Next for $C_{3a}$, note that 
\begin{align*}
	\norm{\tilde{C}_{3a}\Bm}\lesssim\norm{\psi}_\H , \qquad
\norm{	\Bm	\tilde{C}_{3a}}\lesssim\norm{\psi}_\H ,	
\end{align*}
and therefore,
	\begin{align*}
	\Bnorm{\f C_{3a}}{\psi}_\H\lesssim&\f\left(  \Bnorm{\slD\tilde{C}_{3a}}{\psi}_\H + \Bnorm{\tilde{C}_{3a}\slD}{\psi}_\H \right)\\
	\lesssim& \f \norm{\psi}_{\H}.
	\end{align*}
Finally,
	\begin{align*}
\Bnorm{\f \left(a\lambda D_x + \frac{\Lambda a^2}{3} D_\varphi\right)C_{3b}}{\psi}_\H\lesssim \f \norm{C_{3b}\Bm\psi}_\H\lesssim \f \norm{\psi}_\H.
	\end{align*}
 Hence, adding all the terms together, we have shown that
	\begin{equation*}\label{eq:bound on com}
	\Bnorm{\comm{H(\tau)}{B}}{\psi}_{\H} \lesssim \f \norm{\psi}_{\H}.
	\end{equation*}
Integrating over $\tau\in\R$ and recalling that $\f(r(\tau))\in L^1(\R_\tau)$, we get
	\begin{equation*}
	\int_\R \Bnorm{\comm{H(\tau)}{B}}{\psi}_{\H} \diff \tau \lesssim \norm{\psi}_{\H},
	\end{equation*}
as promised.
\end{proof}

We are in position to prove \Cref{prop:slD u bounded}.

\begin{proof}[Proof of \Cref{prop:slD u bounded}]
	We will first show that the evolution system $\mathcal{U}(\tau,s)$ satisfies
	\begin{equation*}
	\norm{B^{1/2}\mathcal{U}(\tau, s) \Bm}_{\B(\H)} \lesssim 1.
	\end{equation*}
	for all $\tau, s \in \R$, from which the claim of the Proposition can be proven. 
	
	To this aim we follow \cite{derezinski_scattering_1997}. Fix $s \in \R$. Let $\varepsilon >0$ and $\psi \in D$. Put
	\begin{equation*}
	k_{\varepsilon}(\tau):= \norm{B^{1/2}\left(1+\varepsilon B\right)^{-1/2} \mathcal{U}(\tau, s) \psi}_{\H}^2.
	\end{equation*}
Note that since $\mathcal{U}(\tau,s)$ is unitary,
\begin{align*}
		k_{\varepsilon}(\tau)&=\frac{1}{\varepsilon}\inner{(-1+1+\varepsilon B)\left(1+\varepsilon B\right)^{-1/2} \mathcal{U}(\tau, s) \psi}{\left(1+\varepsilon B\right)^{-1/2} \mathcal{U}(\tau, s) \psi}_\H \\ &=\frac1\varepsilon\norm{\psi}_{\H}^2 - \frac1\varepsilon\norm{\left(1+\varepsilon B\right)^{-1/2} \mathcal{U}(\tau, s) \psi}_{\H}^2.
\end{align*}
	\Cref{lem:property of B} allows us to take $\tau$ derivatives. Therefore, we have
	\begin{align*}
	\frac{\diff}{\diff \tau} k_{\varepsilon}(\tau) &= - \frac1\varepsilon \frac{\diff}{\diff \tau} \inner{\,\mathcal{U}(s,\tau) \left(1+\varepsilon B\right)^{-1} \mathcal{U}(\tau, s) \psi}{\psi}_\H
	\\
	&= -\frac1\varepsilon \inner{\,\mathcal{U}(s,\tau) \comm{iH(\tau)}{\left(1+\varepsilon B\right)^{-1}} \mathcal{U}(\tau, s) \psi}{\psi}_\H
	\\
	&= \inner{\,\mathcal{U}(s,\tau) {\left(1+\varepsilon B\right)^{-1}} \comm{B}{iH(\tau)} {\left(1+\varepsilon B\right)^{-1}} \mathcal{U}(\tau, s)\psi}{\psi}_\H,
	\end{align*}
where the last equality follows from 
\begin{equation*}
	\comm{B}{H}=\frac{1}{\varepsilon}\comm{1-1+\varepsilon B}{H}.
\end{equation*}
	 Cauchy-Schwarz then gives us
	\begin{align*}\hspace{-0.4cm}
	\abs{\frac{\diff}{\diff \tau} k_{\varepsilon}(\tau)} &\le \abs{\inner{\Bm \comm{H(\tau)}{B} \Bm B^{1/2} {\left(1+\varepsilon B\right)^{-1}} \mathcal{U}(\tau, s)\psi}{B^{1/2} {\left(1+\varepsilon B\right)^{-1}} \mathcal{U}(\tau, s)\psi}_\H}
	\\
& \le  \norm{\Bm \comm{H(\tau)}{B} \Bm}_{\B(\H)}\norm{(1+\varepsilon B)^{-1}}_{\B(\H)} k_{\varepsilon}(\tau) 
\quad \forall \tau \in \R.
	\end{align*}
 Hence, by Grönwall's lemma evaluated at $s$,
	\begin{align}
	k_{\varepsilon}(\tau) &\leq \exp\left(\norm{(1+\varepsilon B)^{-1}}_{\B(\H)}\int_{s}^{\tau}\norm{\Bm \comm{H(\zeta)}{B} \Bm}_{\B(\H)} \diff \zeta\right) k_\varepsilon(s) \nonumber\\ &\lesssim  \norm{B^{1/2}\left(1+\varepsilon B\right)^{-1/2} \psi}_{\H}^2, \label{eq:Gronwallconstant}
	\end{align}
by \Cref{lem:H-B com} and the fact that $\mathcal{U}(s,s)=\mathrm{Id}_\H$. Now as $\varepsilon \to 0$, we have
	\begin{align*}
	\norm{B^{1/2}\mathcal{U}(\tau, s) \psi}_{\H} \lesssim \norm{B^{1/2} \psi}_{\H}.
	\end{align*}
	From the fact that $\Bm$ is a bijection from $\H$ onto $D=D(B^{1/2})$, we have shown that for all $\tilde\psi\in\H$,
	\begin{align*}
	\norm{B^{1/2}\mathcal{U}(\tau, s) \Bm \tilde\psi}_{\H} \lesssim \norm{ \tilde\psi}_{\H}.
	\end{align*}
Using this, together with \Cref{lemma:B norm} and \eqref{eq:norm equiv Bm}, we have
	\begin{align*}
	\norm{\slD \, \mathcal{U}(\tau, s) \psi}_{\H} 
	= \norm{\slD \Bm B^{1/2} \, \mathcal{U}(\tau, s) \psi}_{\H}
	&\lesssim \norm{ B^{1/2} \, \mathcal{U}(\tau, s) \Bm B^{1/2} \psi}_{\H}
	\\
	&\lesssim \norm{B^{1/2} \psi}_{\H} \eqsim \norm{\psi}_{\Hsp{1}(\Sigma)},
	\end{align*}
proving \Cref{prop:slD u bounded}.
\end{proof}

\appendix

\section{Conditions on the KN(A)dS parameters}\label{appendix:Config}

We give the conditions on the free parameters $M,Q,a$ and $\Lambda$ of the KN(A)dS family of metrics equivalent to fulfilling Hypotheses \ref{hypo:1}, namely, parts (\textit{ii}), (\textit{iv}) and (\textit{v}) of \Cref{sec:configurations}, the rest have been fully discussed in that section.

%

Recall,
\begin{equation}\label{horizonfntinAppendix}
	\Delta_r(r)=\left(1-\frac{\Lambda r^2}{3}\right)(r^2+a^2)-2Mr+\left(1+\frac{\Lambda a^2}{3}\right)^2Q^2.
\end{equation}
Suppose first $\Lambda\le0$. We may immediately disregard the $M<0$ case as it renders $\Delta_r \ge 0$ for all $r \in \R^+$ which is in conflict with \textit{(h1)} of \Cref{hypo:1}. What is left to analyze are the cases $$(M >0 , \Lambda < 0) \qquad \text{and} \qquad (M \ne 0 , \Lambda > 0).$$ Set $\Lambda = 3\epsilon/\ell^2$ with $\epsilon = \pm 1$ depending on the sign of $\Lambda$. Now set $K := a^2 + \lambda^2 Q^2 \ge 0$ and
\begin{equation}
	P(r) := -\epsilon\ell^2 \Delta_r(r) = r^4 + Ar^2 + Br + C	
\end{equation}
with
\begin{align*}
	A = A(\epsilon) &= \ell^2 - \epsilon a^2, & A(-1) > 0,\\
	B = B(\epsilon) &= 2 \epsilon M\ell^2, & B(-1) < 0,  && B(+1) \ne 0,\\
	C = C(\epsilon) &= - \epsilon \ell^2 K, & C(-1) \ge 0, && C(+1) \le 0.
\end{align*}
The discriminant $\Delta_P$ of $P(r)$ is given by
\begin{equation}\label{discrDeltaP 2}
	\Delta_P=-27B^4+A(144C-4A^2)B^2+256C^3-128A^2C^2+16A^4C,
\end{equation}
which may be rewritten as a quadratic polynomial on $X := B^2$ as
\begin{equation}
	\Delta_P(X):=\Delta_P=-27X^2+\beta X+\gamma ,
\end{equation}
with
\begin{align*}
	\beta = \beta(\epsilon) &= 4A(36C-A^2),\\
	\gamma = \gamma(\epsilon) &= 16C(A^2-4C)^2, & \gamma(-1) \ge 0, && \gamma(+1) \le 0.
\end{align*}
It is then a straightforward calculation to see that the discriminant of $\Delta_P(X)$ is 
\begin{equation}\label{eq:discriminant of discriminant}
	\delta = \delta(\epsilon) = 16(A^2+12C)^3,
\end{equation}which is readily seen to be positive when $\epsilon = -1$. 

It is a well-known result (see, e.g. \cite{rees_graphical_1922}) that the three possible cases regarding the sign of $\Delta_P$ are:
\begin{enumerate}
	\item $\Delta_P<0$: the roots are simple, two real and one complex conjugate pair;
	\item $\Delta_P>0$: the roots are simple:
	\begin{enumerate}
		\item $A<0 \text{ and } C<\frac{A^2}{4}$: four  real roots;
		\item $A\ge 0 \text{ or } C>\frac{A^2}{4}$: two complex conjugate pairs;
	\end{enumerate}
	\item $\Delta_P=0$: at least one double root:
	\begin{enumerate}
		\item $A<0 \text{ and } -\frac{A^2}{12}<C<\frac{A^2}{4}$: the roots are real with two simple and one double;
		\item $A<0 \text{ and } C=\frac{A^2}{4}$: two double real roots;
		\item $A<0 \text{ and } C=-\frac{A^2}{12}$: the roots are real with one simple and one cubic;
		\item $A>0 \text{ and } C=\frac{A^2}{4} \text{ and } B=0$: two double roots, conjugates of each other;
		\item $\left(A>0 \text{ and } \left(C\ge0 \text{ and } B \neq 0\right)\right) \text{ or }\left(A=0 \text{ and } C>0\right)$ or \\ $\left(A<0 \text{ and } 4C>A^2\right) $: one double real root and one complex conjugate pair;
		\item $A=0 \text{ and } C=0$: a quartic real root.
	\end{enumerate}
\end{enumerate}

Now note that in the case $\Lambda >0$, it is necessary for $\Delta_r$ to have three distinct positive roots for \Cref{hypo:1} to hold due to the negative asymptotics of $\Delta_r$, and therefore all four roots must be real. While for $\Lambda<0$, \Cref{hypo:1} can, a priori,  hold if the largest two positive roots are simple or if there is a third larger double root.  
From here, in the case $\Lambda < 0$, the only root configurations suitable to our hypotheses are 1., 2.(a) and 3.(a). However, we have $A(-1)>0$ which immediately narrows our analysis to 1. For $\Lambda > 0$, 2.(a) is the only case compatible with our hypotheses. Note that $ 4C(+1) < A^2(+1) $ already holds as $C(+1) \le 0$. Hence, the requirements for \Cref{hypo:1} are
\begin{align}\label{eq:force}
	(i) \quad\Delta_P < 0 \quad \text{for}\quad \Lambda < 0; \qquad \text{and} \qquad &(ii)\quad \Delta_P > 0 \quad \text{and} \quad A < 0, \quad \text{for} \quad \Lambda > 0 .
\end{align}
Denote by $X_\pm$ the roots of $\Delta_P(X)$ whose expressions are given explicitly by
\begin{equation}\label{eq:Xpm}
	X_\pm=\frac{2}{27}\left(-A^3+36AC\pm (A^2+12C)^{3/2}\right).
\end{equation}
From the conditions in \eqref{eq:force} on the discriminant $\Delta_P$, we see that $X_\pm \in \R$ since in (\textit{i}), $\delta(-1)>0$, while in (\textit{ii}), $\Delta_P<0$ otherwise. Hence, (\ref{eq:force}) becomes
\begin{equation}\label{eq:force 2}
	\begin{aligned}
		(i)&\quad X < X_- \quad \text{or} \quad X > X_+  \quad \text{for}\quad \Lambda < 0, \\
		(ii) &\quad \delta(+1)>0 \quad \text{and}\quad X_- < X < X_+ \quad \text{for} \quad \Lambda > 0 .
	\end{aligned}
\end{equation}

By Vieta's formulae for the quadratic polynomial $\Delta_P(X)$, we have that 
\begin{equation}\label{eq:vieta Xpm}
	X_- + X_+ = \frac{\beta}{27}  \quad \text{and} \quad X_- X_+ = -\frac{\gamma}{27},
\end{equation}
i.e., $X_-X_+\le 0$ for $\epsilon = -1$, therefore $X_- \le 0$ and so $X<X_-$ is not attained because $X=B^2\in\R^+$. A similar argument for $\epsilon = +1$ yields that $X_+ > X_- \ge 0$.


\begin{lemma}[$\Lambda>0$]\label{lem:Lambda pos vieta}
	Let $P(r)=r^4+Ar^2+Br+C$  with $A\in\R$, $B>0$ and $C<0$, and assume that its discriminant $\Delta_P > 0$. Then exactly one root is negative and the others are positive. 
\end{lemma}
\begin{proof}
	Let $r_1,r_2,r_3,$ and $r_4$ be the roots of $P(r)$, which are real and simple by 2.(a) since $C<0$. By Vieta's formulae: 
	\begin{align}
		&r_1+r_2+r_3+r_4=0,\label{Vieta1}\\
		&r_1r_2r_3+r_1r_2r_4+r_1r_3r_4+r_2r_3r_4=-B,\label{Vieta2}\\
		&r_1r_2r_3r_4=C. \label{Vieta3}
	\end{align}
	Put $p=r_3r_4$ and $s=r_3+r_4$. From \eqref{Vieta3} we can assume without loss of generality that $r_1<0$ and $r_2>0$, and thus $p>0$ and $r_1r_2=C/p$, and from \eqref{Vieta1}, $r_1+r_2=-s$. Therefore \eqref{Vieta2} becomes:
	\begin{equation}\label{eq:Vieta cons}
		\frac{C}{p}s-ps=-B,
	\end{equation}
	and hence $s>0$, proving the claim.
\end{proof}
\begin{remark}
	From the proof of the previous lemma, we recall that the case $M<0$ in de Sitter yields a positive $B$. Then \eqref{eq:Vieta cons} forces $s<0$, i.e. three out of four real roots are negative. Thus we disregard the case $M<0$.
\end{remark}

{
	\begin{lemma}[$\Lambda<0$]\label{lem:Lambda neg vieta}
		Let $P(r)=r^4+Ar^2+Br+C$  with $A>0$, $B<0$ and $C>0$ and assume its discriminant $\Delta_P < 0$. Then $P(r)$ admits two simple \underline{positive} roots and a pair of complex conjugated roots.
	\end{lemma}
	\begin{proof}
		Since $\Delta_P$ < 0, it immediately follows that $P(r)$ admits two distinct real roots ($r_- < r_+$, say) and two complex conjugated roots ($z$ and $\bar z$, say). Then, Vieta's formulae yield
		\begin{align}
			&r_++r_-+z+\bar z=0,\label{Vieta 1}\\
			&r_-r_+ + r_-z + r_- \bar z + r_+ z + r_+ \bar z + z \bar z = A, \label{Vieta 2} \\
			&r_+r_-z+r_+r_-\bar z+r_+z\bar z+r_-z\bar z=-B,\label{Vieta 3}\\
			&r_+r_-z\bar z=C. \label{Vieta 4}
		\end{align}
		We write $z := x + iy$ (note that $y\ne0$) and set $s:=r_+ + r_-, p:= r_-r_+$. Then Vieta's formulae imply
		\begin{align*}
			&2x + s = 0, \\
			&2xs + \abs{z}^2 + p > 0, \\
			&\abs{z}^2 s + 2xp > 0, \\
			&\abs{z}^2p > 0,
		\end{align*}
		and thus,
		\begin{align}
			&-s^2 + \abs{z}^2 + p > 0, \label{eq:v1} \\
			&s(\abs{z}^2-p)> 0, \label{eq:v2}\\
			&p > 0.	\label{eq:v3}
		\end{align}
		Note that $s^2 = r_+^2 + 2p + r_-^2 > 2p$ since $r_\pm \ne 0$ by \eqref{eq:v3}. Hence $p-s^2 < -p$. Coupled with \eqref{eq:v1}, this inequality yields $\abs{z}^2-p > 0$, whence $s>0$.
	\end{proof}
}
For \Cref{lem:Lambda pos vieta,,lem:Lambda neg vieta}, we need $C\ne0$, hence, we do not allow $K=0$.\footnote{Note that $K=0$ corresponds to a Schwarzschild-type black hole.}
In summary, we  have shown that:
\begin{proposition}\label{prop:distinct horizons}
	If $M>0, \Lambda \in \R$ and $K>0$, then the horizon function $\Delta_r(r)$ satisfies Hypothesis \ref{hypo:1} if and only the following conditions on the free parameters hold:
	\begin{itemize}
		\item For $\Lambda < 0$
		\begin{align}
			& 36M^2 > \Lambda^2 X_+ .\label{Cond}
		\end{align}	
	\end{itemize}
	\begin{itemize}
		\item For $\Lambda > 0$
		\begin{align}
			&\Lambda^2 X_- < 36M^2 < \Lambda^2 X_+, \label{Cond3} \\ 
			& Q^2 < \frac{\Lambda^2a^4-42\Lambda a^2+9}{4\Lambda(\Lambda a^2+3)^2}, \label{Cond1} \\
			& \Lambda a^2 < 3.  \label{Cond2}
		\end{align}
	\end{itemize}
	
\end{proposition}
\begin{remark}
	We note that when $Q=0$, conditions \eqref{Cond1} and \eqref{Cond2} reduce to $a\sqrt{\Lambda}<\sqrt{3}(2-\sqrt{3})$, which is the way they are stated in \cite{borthwick_maximal_2018}. While when $a=0$ they reduce to $4\Lambda Q^2<1$, and we recover the condition originally found\footnote{In \cite{mokdad_reissner-nordstrom-sitter_2017} $\Lambda/3$ was replaced by $\Lambda$.} in \cite{mokdad_reissner-nordstrom-sitter_2017}. Similarly, condition \eqref{Cond3} agrees with the previous works \cite{borthwick_maximal_2018} (for $Q=0$) and \cite{mokdad_reissner-nordstrom-sitter_2017} (for $a=0$). Finally, when $a=0$, \eqref{Cond} with $K>0$ agrees with the condition given in\footnote{Again, in \cite{hafner_scattering_2021} $\Lambda/3$ was replaced by $\Lambda$.} \cite{hafner_scattering_2021}.
\end{remark}

\section{Spin coefficients}\label{sec:spin coeffs}
The spin coefficients \eqref{eq:spin coeffs} used throughout the paper are obtained in this section. As far as we are aware of, this is a novel calculation for cosmological Kerr spacetimes, and certainly for the tetrad $\{{e_i}^a\}_{i=1,2,3,4} = \mathbb{T}$ defined in \eqref{eq:tetrad}. In what follows, we will use $i, j, k$ for tetrad indices and $a, b, c$ for vector indices. Denote by
\begin{equation}\label{eq:lambda}
	\lambda_{ijk} = -\lambda_{kji} = \big(e_{j a,b} - e_{j b,a}\big) {{e_{i}}^a} {{e_{k}}^b}
\end{equation}
the connection coefficients $\lambda_{ijk}$ which then define the Ricci rotation coefficients
\begin{equation}\label{eq:gamma}
	\gamma_{ijk} = 
	\frac12 \bigg(
	\lambda_{ijk} + 
	\lambda_{kij} - 
	\lambda_{jki}
	\bigg).
\end{equation}
In the calculations to come, we repeatedly make use of the vector fields
\begin{gather*}
	v := \left(r^2+a^2\right) \partial_t + a \partial_\varphi = \frac{\Delta_r}{ \lambda^{2}} \alpha^\sharp\,, \\ w := \partial_t + \frac{1}{a\sin^2\theta} \partial_\varphi = \frac{\Delta_\theta}{a \lambda^{2}}  \beta^\sharp\,,
\end{gather*}
where, evidently, $\g(v, w) = 0$. Also, we shall denote for convenience
\begin{equation}\label{eq:h}
h(r,\theta):= \frac{-\Delta_r(r)}{\rho^2} = - \frac{\lambda (r^2+a^2)}{\rho^2} f(r), \quad h' := \partial_r h.
\end{equation}
We provide the calculations of $\lambda_{112}$, $\lambda_{113}$, $\lambda_{114}$ and $\lambda_{314}$ to illustrate the relevant computational steps.
\begin{align*}
\lambda_{112}
&= 
\big(\partial_b l_{a} - \partial_a l_{b}\big) l^a n^b
= 
\frac{-\Delta_r}{2\rho^2}
\big(\partial_b l_{a} 
- 
\partial_a l_{b}\big)
\overbrace{
	\bigg(\partial_r^a + \frac{\lambda}{\Delta_r} v^a\bigg)\bigg(\partial_r^b - \frac{\lambda}{\Delta_r} v^b\bigg)
}^{\text{symmetric part vanishes}}
\\
&= 
-\frac{\lambda}{2\rho^2}
\big(\partial_b l_{a} - \partial_a l_{b}\big) \bigg({\delta_1}^b v^a - {\delta_1}^a v^b\bigg) 
= 
-\frac{\lambda}{\rho^2}
\big(v^a \partial_r l_{a} - \overbrace{v^a\partial_a l_{1}}^{=0}\big)
\\
&=
\frac{\lambda}{\rho^2}
v^a \partial_r \bigg( \frac{{\delta_a}^1 + \frac{\lambda}{\Delta_r}v_a}{\sqrt{2 g_{rr}}}\bigg) 
=
\frac{1}{\rho^2}
v^a \partial_r \bigg( \frac{ \alpha_a}{\sqrt{2 g_{rr}}}\bigg)
=
\frac{\alpha(v)}{\rho^2\sqrt 2} \partial_r (-\rho^2/\Delta_r)^{-\frac12}
\\[10pt]
&=
\frac{\rho}{2\sqrt{-2\Delta_r}} \partial_r \frac{-\Delta_r}{\rho^2}
=
\frac{-1}{2\sqrt{-2\Delta_r}} \frac{\rho^2 \Delta_r' - 2r\Delta_r}{\rho^3}
=
-\frac{\rho^2 \Delta_r' - 2r\Delta_r}{2\rho^3\sqrt{-2\Delta_r}}.
\end{align*}
\begin{align*}
\lambda_{113}
&= 
\big(\partial_b l_{a} - \partial_a l_{b}\big) l^a m^b 
= 
-m^2 \overbrace{l_a \partial_\theta l^a}^{=0} + l^1 \overbrace{l_a \partial_r m^a}^{=0} = 0, \quad \lambda_{114} = \overline{\lambda_{113}} = 0.
\end{align*}
\begin{align*}
\lambda_{314}
&=
\big(\partial_j l_{i} - \partial_i l_{j}\big) m^i \overline m ^j
=
m^2\big(m^i - \overline m^i\big)\partial_\theta l_i 
=
\frac{\sqrt h}{\sqrt 2 \rho} \frac{2ia\lambda\sin\theta}{\sqrt{2h}\rho} \overbrace{w^i \partial_\theta l_i}^{=-l_i\partial_\theta w^i}
\\
&=
-\frac{ia\lambda\sin\theta}{\rho^2} \Bigg(-\sqrt\frac{-\Delta_r}{2\rho^2}\Bigg)\frac{1}{\lambda}\alpha\bigg(-\frac{2 \cos\theta}{a \sin^3\theta} \partial_\phi\bigg) 
=
\frac{ia\sqrt{-2\Delta_r}\cos\theta}{\rho^3}.
\end{align*}
The other spin coefficients are obtained similarly, all by hand. We present only the results of these calculations. The omitted coefficients are either vanishing or can be obtained by the antisymmetric property of $\lambda_{ijk}$ with respect to $i$ and $k$ as in \eqref{eq:lambda}.
\begin{align*}
 & \lambda_{112} = - \lambda_{122} = -\frac{\rho^2 \Delta_r' - 2r\Delta_r}{2\rho^3\sqrt{-2\Delta_r}};
 & \lambda_{123} = \lambda_{124} = -\frac{a^2\sin\theta \cos\theta}{\rho^3} \sqrt{\frac{h}{2}};
 \\
 & \lambda_{132} = -\lambda_{142} = \frac{i ar \sin\theta \sqrt{2h}}{\rho^3};
 & \lambda_{134} = \lambda_{143} = \lambda_{234} = \lambda_{243} = -\frac{r}{\rho^3}\sqrt{\frac{-\Delta_r}{2}};
 \\
 & \lambda_{213} = \lambda_{214} = -\frac{a^2\sin\theta \cos\theta}{\rho^3} \sqrt{\frac{h}{2}};
 & \lambda_{314} = -\lambda_{324} = \frac{ia\sqrt{-2\Delta_r}\cos\theta}{\rho^3};
\end{align*}
\begin{equation*}
	\lambda_{334} = -\lambda_{344} = \frac{\cot\theta}{\rho}\sqrt\frac{h}{2} \,
	\left[1 + \frac{a^2 \sin^2\theta}{h \rho^2}\Big(1 - \frac{\Lambda r^2}{3}\Big)\right].
\end{equation*}

The general formulae for the spin coefficients in terms of the Ricci rotation coefficients are available in, say, \cite[Chapter~10,~Section 102]{chandrasekhar_mathematical_1998} and read
\begin{align*}
\kappa_s &= \gamma_{311}
= 
\frac12 \bigg(\lambda_{311} + \lambda_{131} - \lambda_{113}\bigg)
=
-\lambda_{113}
=
0;
\\
\rho_s &= \gamma_{314}
=
\frac12 \bigg(
\lambda_{314} + 
\lambda_{431} - 
\lambda_{143}
\bigg)
=
\frac12 \bigg(
\frac{ia\sqrt{-2\Delta_r}\cos\theta}{\rho^3}+\frac{r}{\rho^3}\sqrt{\frac{-\Delta_r}{2}}+\frac{r}{\rho^3}{\sqrt{\frac{-\Delta_r}{2}}}
\bigg)
\\
&=
\frac{r + ia\cos\theta}{\rho^3}
\sqrt{\frac{-\Delta_r}{2}};
\\
\epsilon_s &= 
\frac12 \bigg(
\gamma_{211} + \gamma_{341} \bigg)
=
\frac14 \bigg(
\lambda_{211} + \lambda_{121} - \lambda_{112} + \lambda_{341} + \lambda_{134} - \lambda_{413}
\bigg)
\\
&= \frac14
\bigg(
\frac{\rho^2 \Delta_r' - 2r\Delta_r}{\sqrt2 \sqrt{-\Delta_r}  \rho^3}
+
\frac{r}{\rho^3}\sqrt{\frac{-\Delta_r}{2}}
-
\frac{r}{\rho^3}\sqrt{\frac{-\Delta_r}{2}}
+
\frac{ia\sqrt{-2\Delta_r}\cos\theta}{\rho^3}
\bigg)
\\[5 pt]
&=
\frac{\rho^2 \Delta_r' - 2(r + ia \cos\theta) \Delta_r}{4\rho^3\sqrt{-2\Delta_r}};
\\
\sigma_s &= \gamma_{313}
=
\frac12\bigg(
\lambda_{313} + 
\lambda_{331} - 
\lambda_{133}
\bigg)
=
0;
\\
\mu_s &= \gamma_{243}
=
\frac12\bigg(
\overbrace{\lambda_{243}}^{=\overline{\lambda_{234}}} + 
\lambda_{324}  
\overbrace{-\lambda_{432}}^{=+\lambda_{234}}
\bigg)
=
\frac12
\bigg(
-\frac{2r}{\rho^3}\sqrt{\frac{-\Delta_r}{2}}
+
-\frac{ia\sqrt{-2\Delta_r}\cos\theta}{\rho^3}
\bigg)
\\
&=
-\frac{r+ia\cos\theta}{\rho^3} \sqrt{\frac{-\Delta_r}{2}};
\\
\gamma_s
&= 
\frac12 \bigg(\gamma_{212} + \gamma_{342}\bigg)
=
\frac14
\bigg(
\overbrace{
	\lambda_{212} + \lambda_{221} - \lambda_{122}
}
^
{-2\lambda_{122}}
+ 
\overbrace{\lambda_{342}}^{=-\overline{\lambda_{234}}}
+ \lambda_{234} - \lambda_{423}
\bigg)
\\
&=
\frac14
\bigg(
-\frac{\rho^2 \Delta_r' - 2r\Delta_r}{\sqrt2 \sqrt{-\Delta_r}  \rho^3}
-\frac{ia\sqrt{-2\Delta_r}\cos\theta}{\rho^3}
\bigg) 
= 
-\frac{\rho^2 \Delta_r' - 2(r + ia \cos\theta) \Delta_r}{4\rho^3\sqrt{-2\Delta_r}};
\\
\lambda_s &= \gamma_{244}
=
\frac12 \bigg(
\lambda_{244} + \lambda_{424} - \lambda_{442}
\bigg)
=
0;
\\
\tau_s &= \gamma_{312}
=
\frac12 \bigg(
\lambda_{312} + \lambda_{231} - \lambda_{123}
\bigg)
=
\frac12
\bigg(
\frac{2a^2\sin\theta \cos\theta}{\rho^3} \sqrt{\frac{h}{2}}
-
\frac{i ar \sin\theta \sqrt{2h}}{\rho^3}
\bigg)
\\
&=
-ia\sin\theta\,\frac{r + i a \cos\theta}{\rho^3}\sqrt{\frac{h}{2}};
\\
\nu_s &= \gamma_{242}
=
\frac12 \bigg(
\lambda_{242} +
\lambda_{224} -
\lambda_{422}
\bigg)
=
0;
\\
\pi_s &= \gamma_{241}
=
\frac12 \bigg(
\lambda_{241} +
\lambda_{124} -
\lambda_{412}
\bigg)
=
\frac12 \bigg(
\frac{i ar \sin\theta \sqrt{2h}}{\rho^3}
-\frac{2a^2\sin\theta \cos\theta}{\rho^3} \sqrt{\frac{h}{2}}
\bigg)
\\
&=
ia\sin\theta \, \frac{r + i a \cos\theta}{\rho^3}\sqrt{\frac{h}{2}};
\\
\alpha_s &= \frac12 \bigg(\gamma_{214} + \gamma_{344}\bigg)
=
\frac14 \bigg(
\lambda_{214} + \lambda_{421} - \lambda_{142} + \overbrace{\lambda_{344} + \lambda_{434} - \lambda_{443}}^{=2\lambda_{344}}
\bigg)
\\
&=
\frac14 \bigg(
\frac{i ar \sin\theta \sqrt{2h}}{\rho^3}
-
\frac{2\sqrt h \cot\theta}{\sqrt{2}\rho}
\Big[1 + \frac{a^2 \sin^2\theta}{h \rho^2}\Big(1 - \frac{\Lambda r^2}{3}\Big)\Big]
\bigg)
\\
&=
\frac{\sqrt{2h}}{4\rho^3h} \bigg(
iarh\sin\theta - \cot\theta \Big[\rho^2h + a^2 \sin^2\theta\Big(1 - \frac{\Lambda r^2}{3}\Big)\Big]
\bigg)
\\
&=-\frac{\cot\theta}{2 \rho} \sqrt{\frac h2}
+
\frac{a\sin\theta}{2\sqrt{2h}\rho^3}\bigg(
irh - a\cos\theta\Big(1 - \frac{\Lambda r^2}{3}\Big)
\bigg)
\\
&=-\frac{\cot\theta}{2 \rho} \sqrt{\frac h2}
+
\frac{ia\sin\theta}{2\sqrt{2h}\rho^3}\bigg(
r + ia\cos\theta 
\bigg)
\bigg(1 - \frac{i\Lambda a r\cos\theta}{3}\bigg);
\\
\beta_s &= \frac12 \bigg(\overbrace{\gamma_{213}}^{=\overline{\gamma_{214}}} + \gamma_{343}\bigg)
=
\frac14 \bigg(
\overline{\lambda_{214} + \lambda_{421} - \lambda_{142}} + \overbrace{\lambda_{343} + \lambda_{334} - \lambda_{433}}^{=2\lambda_{334} = -2\overline{\lambda_{344}}}
\bigg)
\\
&=\frac{\cot\theta}{2 \rho} \sqrt{\frac h2}
-
\frac{ia\sin\theta}{2\sqrt{2h}\rho^3}\bigg(
r + ia\cos\theta 
\bigg)
\bigg(1 - \frac{i\Lambda a r\cos\theta}{3}\bigg).
\end{align*}
%

\section*{Statements and Declarations}

\subsection*{Funding and/or Conflicts of interests/Competing interests}

M. Mokdad has received partial financial support from the London Mathematical Society through the Atiyah-UK-Lebanon fellowship award 2022-2023. M. Provci did not receive support from any organization for the submitted work.  The authors declare they have no financial interests.

\subsection*{Data Availability}
Data sharing not applicable to this article as no datasets were generated or analyzed during the current study.


\printbibliography[heading=bibintoc] 
\end{document}